%% file: paper.tex
\documentclass[10pt,conference]{IEEEtran}
\IEEEoverridecommandlockouts

\usepackage{amsmath, amsthm, amsfonts, amssymb, amscd, mathtools, bm}
\usepackage{stmaryrd}

\usepackage{algorithmicx}
\usepackage{algorithm}
\usepackage{algpseudocode}
\def\NoNumber#1{{\def\alglinenumber##1{}\State #1}\addtocounter{ALG@line}{-1}}
\algnewcommand{\LeftComment}[1]{\Statex \(\triangleright\) #1}

\def\dvfs{$\mathsf{DVFS\textnormal{-}DNA}$\xspace}
\def\static{$\mathsf{Static\textnormal{-}DNA}$\xspace}

\usepackage{listings}

\usepackage{graphicx}
\usepackage{caption, subcaption}
\usepackage{float}
\usepackage{adjustbox}

\usepackage[dvipsnames]{xcolor}
\usepackage{soul}

\usepackage{textcomp}
\usepackage{latexsym}
\usepackage{xspace}
\usepackage{babel}
\usepackage{url}
\usepackage{xurl}
\usepackage[dvipsnames]{xcolor} 

\usepackage{balance}

\usepackage{lastpage}
\usepackage{enumerate}
\usepackage{enumitem}
\usepackage{multirow}
\usepackage{multicol}
\usepackage{marginnote}
\usepackage{stfloats}

\usepackage{hyperref}

\newtheorem*{remark}{Remark}
\definecolor{gray}{rgb}{0.4,0.4,0.4}
\definecolor{darkblue}{rgb}{0.0,0.0,0.6}
\definecolor{cyan}{rgb}{0.0,0.6,0.6}
\definecolor{backcolour}{rgb}{0.95,0.95,0.92}

\hypersetup{
    colorlinks=true,
    urlcolor=NavyBlue,
    linkcolor=black,
    citecolor=black
}

\lstdefinelanguage{XML}
{
  morestring=[b]",
  morestring=[s]{>}{<},
  morecomment=[s]{<?}{?>},
  stringstyle=\color{black},
  identifierstyle=\color{darkblue},
  keywordstyle=\color{cyan},
  morekeywords={xmlns,version,type}
}

\newcommand{\differential}{{\rm{d}}}

\newenvironment{packeditemize}{
\begin{itemize}
  \setlength{\itemsep}{0.3pt}
  \setlength{\parskip}{2pt}
  \setlength{\parsep}{0pt}
}{\end{itemize}}

\renewcommand{\IEEEauthorrefmark}[1]{\textsuperscript{#1}}

\newtheorem{theorem}{Theorem}

\newtheorem{proposition}{Proposition}

\newcommand{\linh}[1]{{\color{magenta} [\underline{\bf LP:} #1]}}

\newcommand{\reva}[1]{{\color{black} #1\xspace}}
\newcommand{\revg}[1]{{\color{black} #1\xspace}}

\begin{document}

\title{Generative Profiling for Soft Real-Time Systems and its Applications to Resource Allocation}

\author{
  \IEEEauthorblockN{
    Georgiy A. Bondar$^*$\IEEEauthorrefmark{1}\thanks{$^*$Equal contribution.} \quad
    Abigail Eisenklam$^*$\IEEEauthorrefmark{2} \quad
    Yifan Cai\IEEEauthorrefmark{2} \quad
    Robert Gifford\IEEEauthorrefmark{2} \quad
    Tushar Sial\IEEEauthorrefmark{3} \\
    Linh Thi Xuan Phan\IEEEauthorrefmark{2} \quad
    Abhishek Halder\IEEEauthorrefmark{3}
  }
  \\[0.5em]
  \IEEEauthorblockA{
    \IEEEauthorrefmark{1}University of California, Santa Cruz \quad
    \IEEEauthorrefmark{2}University of Pennsylvania \quad
    \IEEEauthorrefmark{3}Iowa State University
  }
}

\maketitle

\begin{abstract}

Modern real-time systems require accurate characterization of task timing behavior to ensure predictable performance, particularly on complex hardware architectures. Existing methods, such as worst-case execution time analysis, often fail to capture the fine-grained timing behaviors of a task under varying resource contexts (e.g., an allocation of cache, memory bandwidth, and CPU frequency), which is necessary to achieve efficient resource utilization. In this paper, we introduce a novel generative profiling approach that synthesizes context-dependent, fine-grained timing profiles for real-time tasks, including those for unmeasured resource allocations. Our approach leverages a nonparametric, conditional multi-marginal Schr\"odinger Bridge (MSB) formulation to generate accurate execution profiles for unseen resource contexts, with maximum likelihood guarantees. We demonstrate the efficiency and effectiveness of our approach through real-world benchmarks, and showcase its practical utility in a representative case study of adaptive multicore resource allocation for  real-time systems.

\end{abstract}

\pagestyle{plain}
\thispagestyle{plain}

\input{intro}

\input{prediction}

\input{accuracy-eval}

\input{usecases}

\input{sched-eval}

\input{related}


\section{Conclusion}

\noindent We have presented a generative stochastic model for nonparametrically generating
context-dependent profiles for tasks in complex real-time systems.
Our solution requires only a small subset of empirical data to 
generate profiles with high accuracy, including for unseen resource contexts. It achieves orders of magnitude  
reduction in measurement time while providing statistical guarantees.
Through case studies and experimental evaluations, we demonstrated the practical application of our approach and its performance benefits in real-world real-time systems.



\newpage

\section*{Acknowledgment}
\noindent 
This work was supported in part by NSF grants CNS-1955670, CNS-2111688 and CCF-2326606.

\balance
{
\bibliographystyle{abbrv}
\bibliography{paper,abby}

}

\newpage
\setcounter{page}{1}
\input{supplementaryRTAS2026}

\end{document}

%% file: intro.tex
\section{Introduction}\label{sec:intro}

\noindent Modern real-time systems execute increasingly sophisticated software on complex hardware architectures. Ensuring timing predictability in these systems relies critically on {\em accurate characterization of task execution under varying resource configurations}. 
On a multicore platform, a task’s execution behavior---such as the rates of cache misses, memory requests and instruction retirement---can be highly sensitive to its {\em resource context}, which includes the amount of shared cache, memory bandwidth, and CPU frequency allocated to it~\cite{acun2019dvfs,dna-rtas21,xu-rtas19}. For instance, as shown in Fig.~\ref{fig:fftwithcontexts}, the PARSEC $\mathsf{fft}$ benchmark exhibits distinct execution phases with varying instruction rates, which change in response to CPU frequency, cache and memory bandwidth allocation. 
Accurately estimating these timing profiles is critical for effective timing analysis and resource allocation in modern real-time systems. Such profiles enable context-aware execution time distributions for probabilistic timing analysis~\cite{kim2005exact,davis2019survey,cai-rtns24,bozhko2021monte} and efficient multi-phase task scheduling~\cite{rasco-emsoft25}. Furthermore, they are central to adaptive multicore resource allocation~\cite{dna-rtas21} and power management~\cite{acun2019dvfs} techniques, which are essential for improving predictability and resource utilization in real-time systems.
\begin{figure}[htbp]
    \centering
    \begin{subfigure}[b]{0.45\textwidth}
        \centering
        \includegraphics[width=\textwidth]{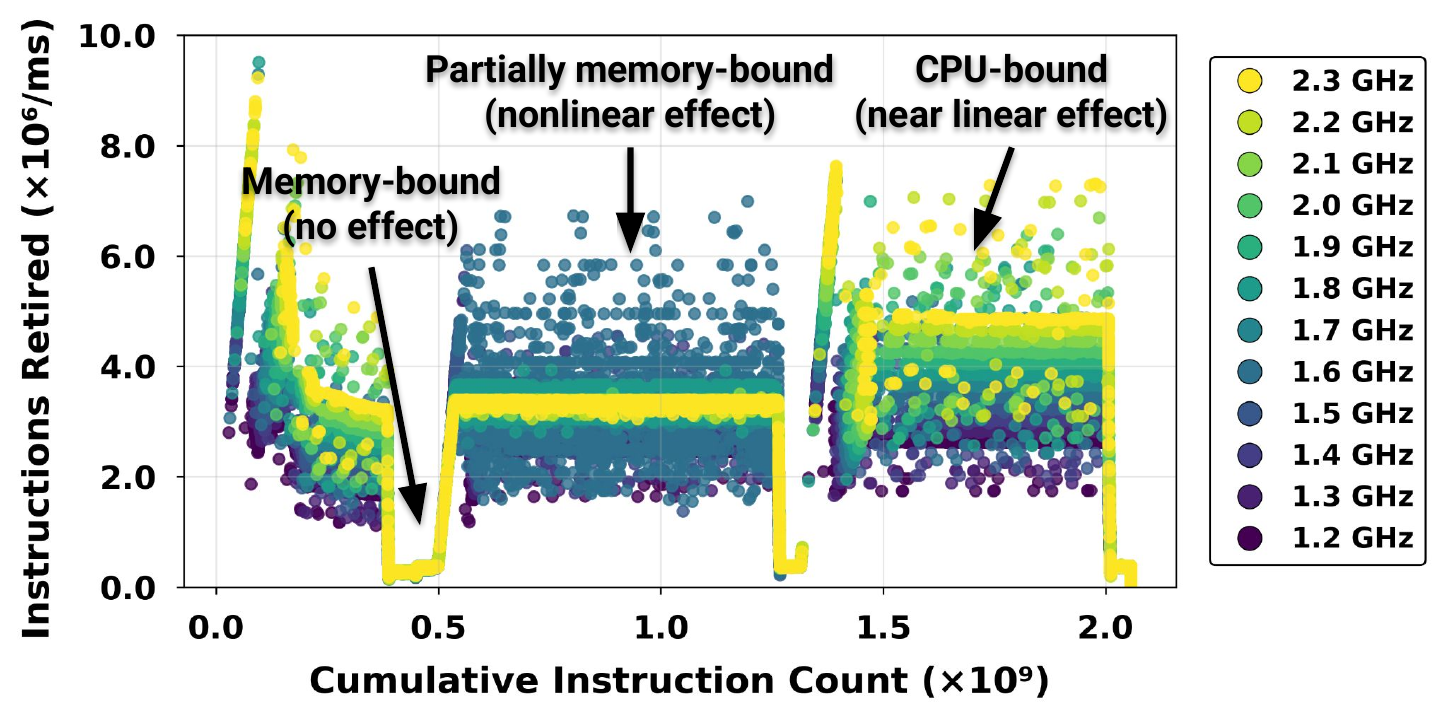}\vspace{-1ex}
        \caption{$\mathsf{fft}$ with 10\% of shared cache and memory bandwidth}
        \label{fig:fft22}
    \end{subfigure}
    \hfill
    \begin{subfigure}[b]{0.45\textwidth}
        \centering
        \includegraphics[width=\textwidth]{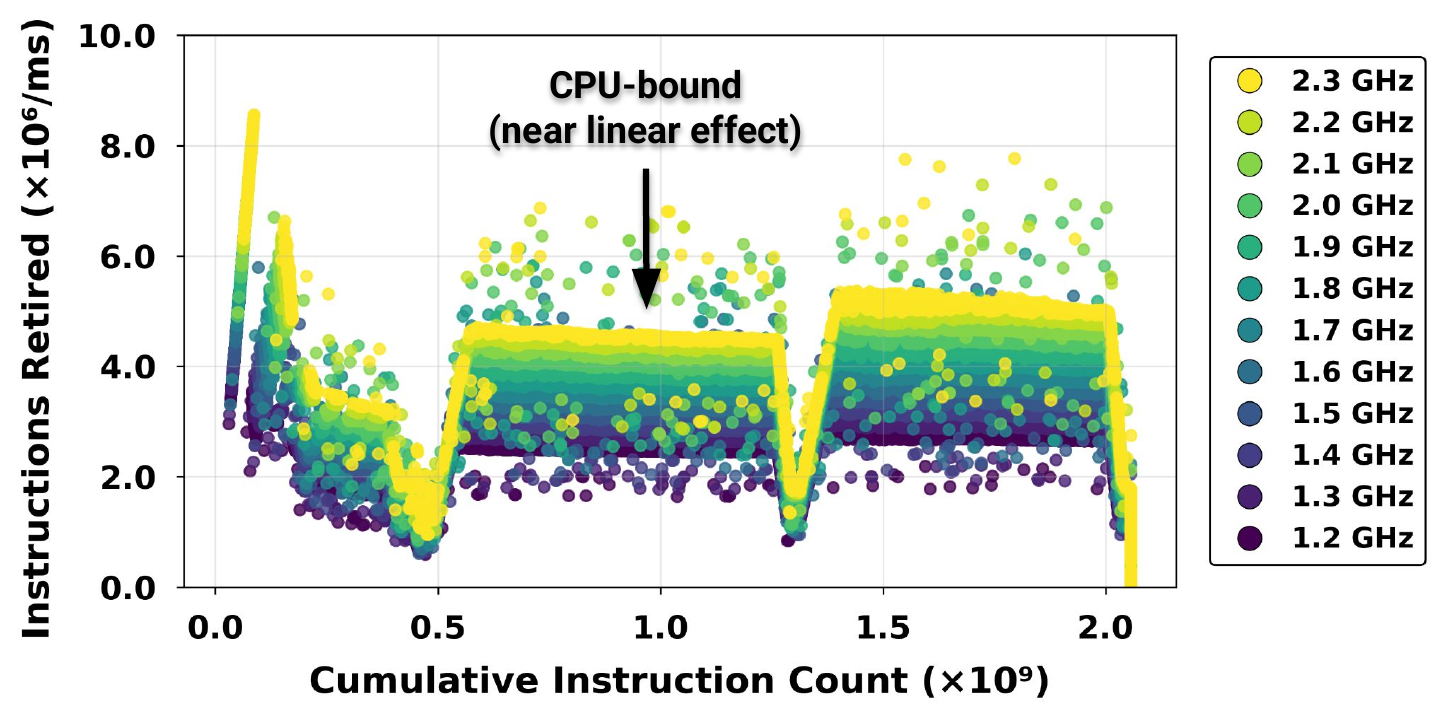}\vspace{-1ex}
        \caption{$\mathsf{fft}$ with 70\% of shared cache and memory bandwidth}
        \label{fig:fft1414}
    \end{subfigure}
    \caption{Effect of CPU frequency, cache, and memory bandwidth on the instruction retirement rate of $\mathsf{fft}$~\cite{splash2x}. The benchmark shows distinct phases, with varying relationships (no effect, linear, nonlinear) between instruction rate and CPU frequency, depending on resource allocation (middle phase).}
    \vspace{-3ex}
    \label{fig:fftwithcontexts}
\end{figure}

\noindent{\bf Limitations of existing work.} 
While traditional analysis methods, such as worst-case execution time (WCET) analysis, can provide bounds or distributions of total execution times, they often fail to capture the fine-grained timing variability under different resource allocations. To obtain fine-grained execution profiles, current solutions have relied on exhaustive measurements. Although this approach can provide fine-grained insights, it can quickly become impractical to cover {\em all possible resource contexts} as the number of tasks or allocation configurations increases. Machine learning-based methods have been proposed to address this limitation; however, they often focus on {\em coarse-grained} execution metrics, such as WCET bounds or aggregated performance distributions. Recently, Bondar et al.~\cite{bondar-2024-stochastic,bondar-2024-psmsbp} applied stochastic learning models to estimate task timing behavior under a given resource context based on measured execution snapshots. However, these methods remain limited to the resource contexts for which data is available and cannot predict behavior in {\em unseen contexts}---resource contexts for which no training data exist.

\noindent{\bf Contributions.}
In this paper, we propose a novel {\em generative profiling} approach to synthesize context-dependent, fine-grained timing profiles for {\em any} resource context, including those that have not been measured. For instance, given \emph{sparse execution snapshots} from a few measured contexts \reva{(e.g., an allocation of 10\% of cache and memory bandwidth and an allocation of 80\% of cache and memory bandwidth)}, our model can infer \emph{complete temporal} profiles for {\em all} resource contexts in between (e.g., an allocation of 20\% of the cache and 50\% of the memory bandwidth). 

\begin{figure}[t]
    \centering
    \includegraphics[width=\linewidth]{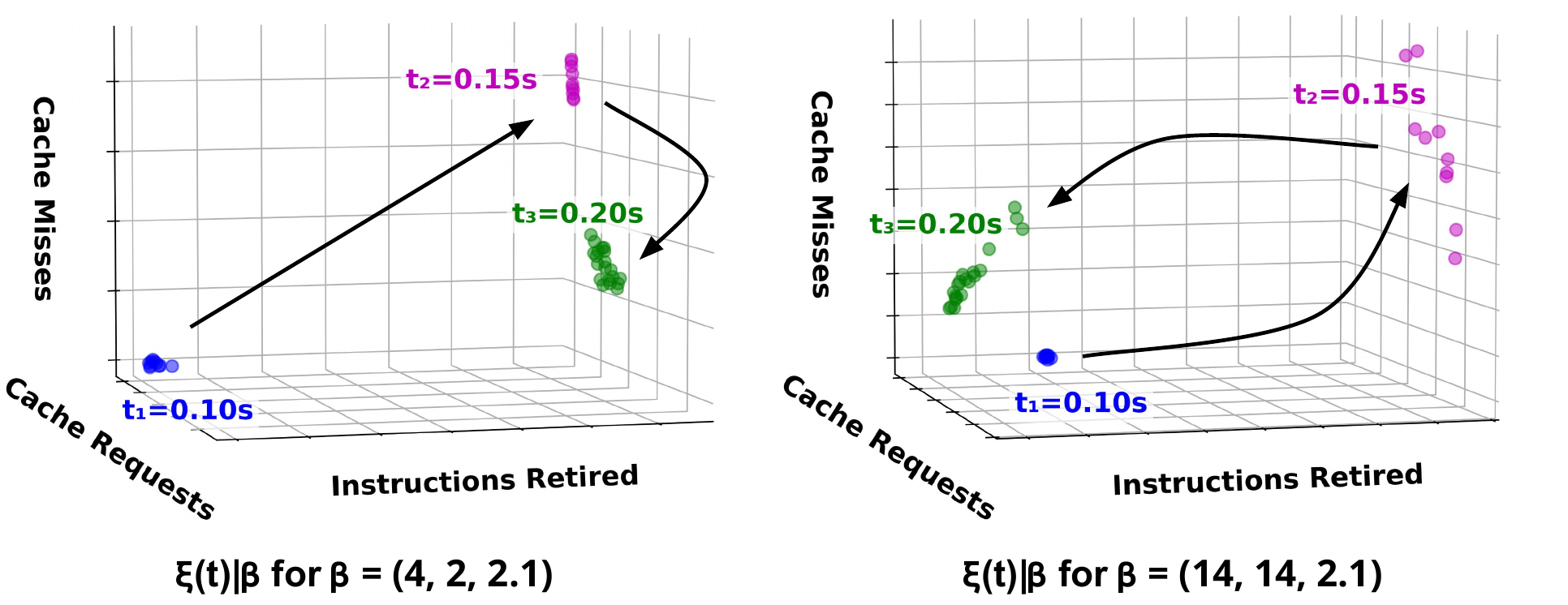}\vspace{-1ex}
    \caption{Distributions of microarchitectural execution states (rates of instructions, cache requests and cache misses) of $\mathsf{fft}$ at three time points under two different resource contexts.}
    \vspace{-3ex}
    \label{fig:dynamic_resource_state}
\end{figure}

Predicting execution profiles for unseen contexts is challenging due to the dynamic nature of task execution and the varying correlations among execution states (e.g., rates of cache misses, memory requests, and instruction retirement). As illustrated in Fig.~\ref{fig:dynamic_resource_state}, these correlations evolve with time and resource context, complicating their modeling with traditional parametric approaches, such as Gaussian mixtures or Markov chains. 
To address this, we introduce a nonparametric learning approach based on a {\em conditional} multi-marginal Schr\"odinger Bridge (MSB) formulation, enabling us to capture complex dependencies from sparse data without making any distributional assumptions.
Our solution efficiently generates accurate fine-grained execution profiles for unseen contexts while providing maximum likelihood guarantees. 

In summary, we make the following contributions:

\begin{packeditemize}
    \item A generative stochastic model for non-parametrically generating context-dependent execution profiles of real-time tasks under unseen contexts, based on a conditional MSB formulation. Our solution is accurate, efficient, and provides maximum likelihood guarantees.

	\item Evaluations of the accuracy and efficiency of our approach using real-world benchmarks.   

	\item A case study demonstrating the practical utility of generative profiles in multicore real-time systems, including an adaptive resource allocation algorithm, its prototype implementation, and experimental results. 

\end{packeditemize}

To the best of our knowledge, this is the first generative profiling solution for real-time multicore systems capable of producing execution profiles for {\em unmeasured resource contexts} with statistical guarantees.

%% file: prediction.tex
\section{Background and Problem Formulation}\label{sec:prediction}

\subsection{Stochastic Modeling of Context-Dependent Profiles} \label{sec:stochastic}
\noindent{\bf System assumptions.} The system contains a set of  tasks that execute on a multicore hardware platform. Tasks running on the cores share a common set of resources, such as the last-level shared cache, memory bandwidth, and main memory.
We assume that each resource can be partitioned into  equal-size allocation units, and that some number of resource units can be assigned to a core at run time (e.g., as done in~\cite{dna-rtas21}). For example, the cache can be efficiently partitioned using methods such as CAT~\cite{intel-2015-cat} (for Intel architecture) or Lockdown by Master (LbM) mechanism~\cite{Mancuso13, arm-pl310} (for ARM architecture), and the memory bandwidth can be partitioned using MemGuard~\cite{Yun16-memguard-journal}. 
The platform is equipped with a per-core power management technique such as Dynamic Voltage and Frequency Scaling (DVFS), which typically supports a range of frequencies at discrete step sizes~\cite{cpufreq-docs}, resulting in a finite number of voltage/frequency settings for a core. For simplicity, we treat frequency as a type of resource.  Thus, a ``resource context'' denotes an assignment of both multicore shared resources {\em and} a voltage/frequency setting to a core. 

\reva{In this work, we focus on tasks with a single execution path, but our generative profiling method is entirely data-driven and can also be applied to general programs. We leave the interpretation and evaluation of generative profiles for multi-path programs as future work.}


\noindent{\bf Definitions.} Formally, the system contains $b$ resource types, where $b$ is some fixed positive integer. A \emph{resource context} $\beta$ is then defined as a vector with $b$ components:
\begin{align}
\beta = \left(\beta_1,\beta_2, \hdots, \beta_b\right)^{\top}\in\mathcal{B}\subset\mathbb{R}^b,
\label{defbeta}    
\end{align} 
where $\beta_i$ denotes the allocation of the $i$-th resource type. 
For example, $b = 3$ resource types may comprise of the shared cache ($\beta_1$), memory bandwidth ($\beta_2$), CPU frequency ($\beta_3$). Then, $\beta = (2, 4, 2.3)$ denotes an allocation of $2$ cache partitions, $4$ bandwidth partitions, and frequency of $2.3$ GHz. 

For some fixed positive integer $m$, we define the \emph{microarchitectural execution state} as an $m$ component vector
\begin{align}
\xi = (\xi_1, \xi_2, \dots, \xi_m)^{\top}\in\mathcal{X}\subset\mathbb{R}^m.
\label{defxi}    
\end{align}
An example of $m = 3$ components of $\xi$ are the number of retired instructions ($\xi_1$), cache requests ($\xi_2$), and cache misses ($\xi_3$) in some unit time interval.

For a fixed resource context $\beta$, the  execution state $\xi$ varies with time $t$. \emph{Mathematically, $\xi(t)$ is a vectorial stochastic process conditioned on a resource context $\beta$}. We denote this conditional stochastic process as 
\begin{align}
\xi(t)\mid\beta.
\label{DefConditionalProcess}\end{align}

The stochastic process \eqref{DefConditionalProcess} is correlated both {\em temporally} (across $\xi$ vectors at different times) and {\em spatially} (across different components of $\xi$ at the same time). Intuitively, if the task issues more shared cache requests at some time $t$, it may issue fewer requests in some future time $t' > t$, as the data may still be resident in the private L1 or L2 cache. Similarly, for any given time $t$ during the task's execution, the instruction rate is correlated to the rates of cache requests and misses, with the correlation depending on the nature of the workload (e.g., compute-intensive or  memory-intensive). Even when $\beta$ and the initial condition $\xi(t=0)$ are fixed \reva{for a deterministic single-path task, repeated executions of that task} under the same resource context results in slightly different sample paths of $\xi(t)$ because of variability in the hardware state and OS behaviors. 

Finally, \emph{profiles} are formally defined as the $\mathbb{R}^{m}$-valued sample paths of the conditional stochastic process \eqref{DefConditionalProcess}.


\begin{figure*}[t]
    \centering
    \includegraphics[width=.68\linewidth]{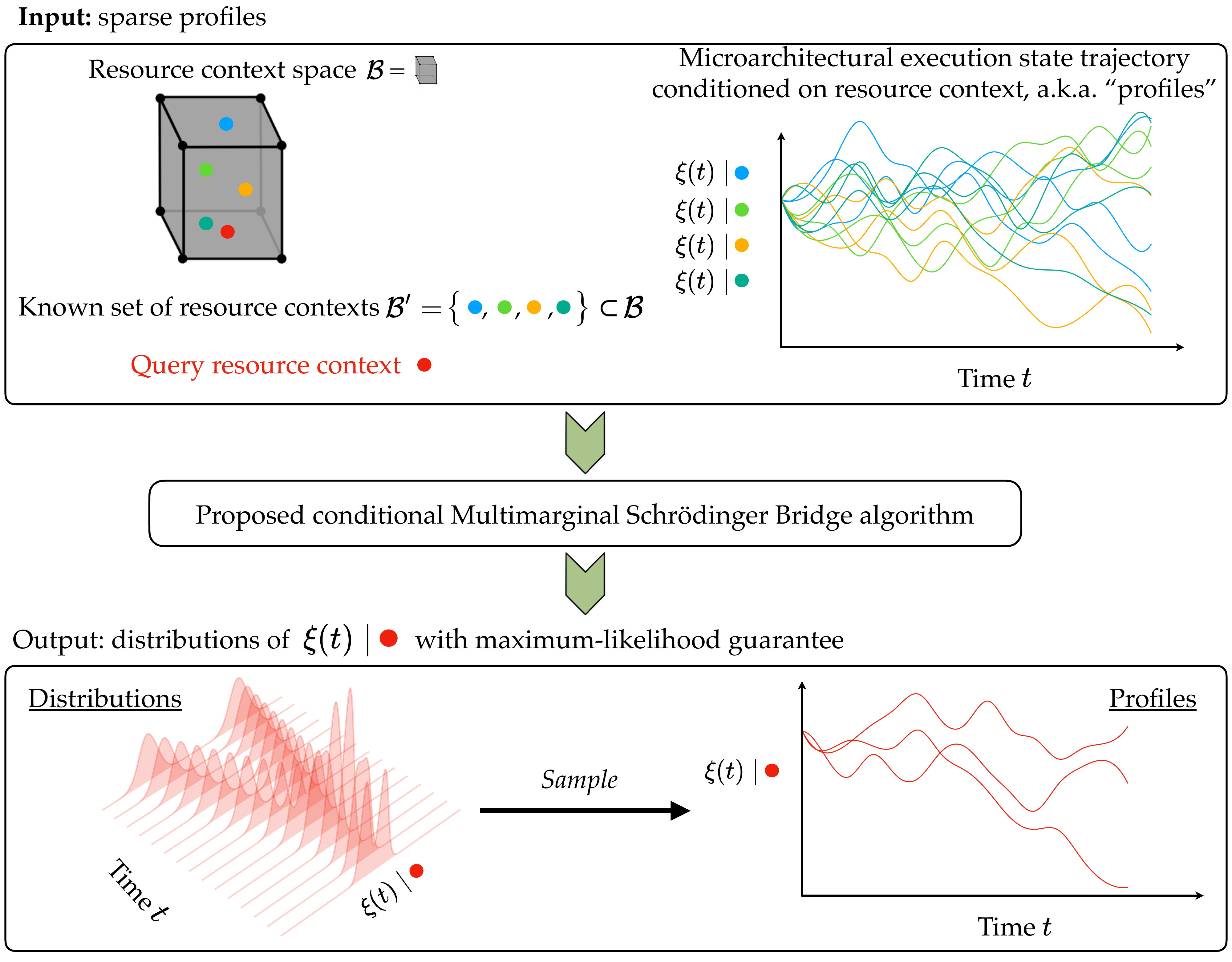} \vspace{-1ex} 
    \caption{Our proposed generative profiling algorithm scheme with maximum-likelihood guarantee. From empirical profiles of a task's microarchitectural execution states, conditioned on a sparse set of resource allocations $\mathcal{B}'$, we generate `synthetic' profiles for the unknown resource allocations.
    }  \vspace{-3ex}
    \label{fig:conceptschematic}
\end{figure*}

\subsection{Nonparametric Learning}\label{sec:whysb}

\noindent Consider the special case of a single resource context $\beta$, i.e., where the set of all possible resource contexts $\mathcal{B}:=\{\beta\}$ is singleton. For this fixed $\beta$, suppose we observe the distributions\footnote{In practice, scattered point clouds, as in Fig.~\ref{fig:dynamic_resource_state}.} of $\xi(t)\vert\beta$ at two fixed times $t_1$ and $t_2$, i.e., we have the snapshots of the execution state at those times for multiple runs of the task. Let us call these snapshots $\mu_1$ and $\mu_2$, respectively. With this information, we desire to learn distributions $\mu_t$ to estimate the statistics of the execution state \emph{at all times} $t\in[t_1,t_2]$. In other words, what we desire to learn is not a single distribution, but a \emph{distribution-valued path} $t\mapsto\mu_t$ parameterized by $t\in[t_1,t_2]$.

Naturally, we must enforce $\mu_{t_1}=\mu_1$ and $\mu_{t_2}=\mu_2$ \reva{to guarantee that the model is \emph{consistent} with the known observations}. For $t\in(t_1,t_2)$, we also want $\mu_t$ to have some mathematical guarantee of its predictive accuracy.
Since we do not make any assumption about the nature of the distribution $\mu_t$, how it evolves with time, or how it is affected by the resource allocation $\beta$, our learning should be nonparametric in nature.


This is where the bimarginal Schr\"odinger Bridge Problem (SBP) and its solution -- the Schr\"odinger Bridge (SB) -- comes in. In the stochastic process literature, the word ``bridge'' \cite[Ch. IV.40]{rogers2000diffusions} refers to a stochastic process whose sample path connects two given \emph{finite dimensional endpoints} at two fixed times $t_1$ and $t_2$. The SB is a similar construction except that the endpoints are now \emph{distributions}. Specifically, among all continuous curves $t\mapsto\mu_t$ connecting the given endpoints $\mu_{t_1}=\mu_1$ and $\mu_{t_2}=\mu_2$, the SB is the one with \emph{maximum likelihood} guarantee, i.e., the highest probability curve among all possible distribution-valued continuous curves connecting the $\mu_1,\mu_2$. The SB originated in 1931-1932 in the works of Schr\"{o}dinger \cite{schrodinger1931umkehrung,schrodinger1932theorie} and in recent years, its importance is being rapidly recognized in generative AI \cite{de2021diffusion,shi2022conditional,liu2022deep,xie2024bridging} and stochastic control \cite{leonard2014survey,chen2021stochastic,caluya2021wasserstein,nodozi2023neural}.

Two mathematical guarantees make the predictive learning model we desire to be most naturally represented by an SB\reva{, rather than another learning model, such as a neural network or regression trees.} \emph{First}, the SB is consistent with known observations $\mu_1$ and $\mu_2$. \emph{Second}, unlike other generative models such as flow matching \cite{lipman2023flow,esser2024scaling}, the SB comes with the most parsimonious statistical guarantee: maximum likelihood certificate. It is appropriate here for, as mentioned, 
we make no assumption about the true mathematical structure of the stochastic dynamics 
underlying the platform on which our tasks run. Importantly, it is also possible to compute the SB in a nonparametric manner, without making any assumptions on the structure\footnote{For example, mixture of Gaussians, exponential family \cite[Ch. 2]{amari2016information} etc.} of $\mu_1,\:\mu_2,$ or $\mu_t$.

In this work, we will employ the multimarginal extension (see Section~\ref{sec:generative}) of SB for the learning of $\mu_t$. We also address the additional complexity that arises when $\mathcal{B}$ is a non-singleton set -- in this case, $\beta$ itself must be treated as a random vector with its own distribution supported over $\mathcal{B}$, statistically correlated with $\xi(t)$.
\reva{Addressing the case in which $\beta$ is not fixed is critical, as it allows us to generate profiles for previously unseen resource contexts and thereby model the fine-grained, context-dependent behavior of tasks without exhaustive measurements across all contexts.}
We describe our learning algorithm next. 



\vspace{-1ex}
\section{Generative Profiling Algorithm}\label{sec:generative}
\noindent  Recall that the set $\mathcal{B}$ in \eqref{defbeta} denotes the set of all possible resource contexts. For instance, let $b=3$ with a total of $N_{\mathsf{ca}}$ cache partitions, $N_{\mathsf{bw}}$ memory bandwidth partitions, $N_{\mathsf{freq}}$ possible settings of CPU frequency. Then, {$\mathcal{B}\subset\mathbb{R}^{3}$ is a finite set of cardinality $N_{\mathsf{ca}}\cdot N_{\mathsf{bw}} \cdot N_{\mathsf{freq}}$.} 

For {\em any}
$\beta\in\mathcal{B}$, we desire a profile $\xi(t) \mid \beta$. However, if $\mathcal{B}$ is large, it may be feasible to obtain these profiles empirically for only a subset $\mathcal{B}'\subsetneq\mathcal{B}$. Moreover, technical limitations in the profiling mechanism or the task itself may lead to these `profiles' consisting of a temporally coarse set of $n_{s}$ snapshots. Therefore, the result of this limited profiling is an empirical distribution of the conditional random vector
\begin{align}
\xi(t_\sigma)\mid\beta
\label{defSnapshotDistribution}  
\end{align} 
for each measurement instance (snapshot) $t_{\sigma\in\llbracket n_s\rrbracket}$,
and for each resource context $\beta\in\mathcal{B}'$. Hereafter, we adopt the finite set notation for the set of $n_s$ snapshots taken for a given task:
\[\llbracket n_s\rrbracket := \{1,2,\hdots,n_s\}.\]
In other words, \eqref{defSnapshotDistribution} is the state $\xi$ at time $t_\sigma$ conditional on the resource context $\beta$.

For a \emph{fixed} resource context $\beta$, recent work \cite{bondar-2024-psmsbp,bondar-2024-stochastic} computed the most likely distributional path $\mu_t$ for $t\in[t_1,t_{n_s}]$ from snapshots $\mu_{\sigma\in\llbracket n_s\rrbracket}$. That is, for a given task and resource context, prior work filled in the profiling gaps between discrete snapshots in time. In this work, we extend these methodologies to not just fill in gaps in time, but also to {\em fill in gaps between the sparse set of resource contexts} $\mathcal{B}'$. 
Specifically, we generate distributions (and thereby, profiles) for each context $\beta \in \mathcal{B} \setminus \mathcal{B}'$ under which the workload was never explicitly measured.

For the generation of these distributions, we will make use of the Multimarginal Schr\"odinger Bridge Problem (MSBP). Different from prior work, we introduce the \emph{conditional} MSBP, explain its solution, and how we use the same to generate profiles for all possible resource contexts. Fig. \ref{fig:conceptschematic} shows the overall schematic of the proposed generative profiling algorithm. Our approach takes as inputs empirical resource usage profiles of a task for a set of `known' contexts $\mathcal{B}'\subset\mathcal{B}$, where  $\mathcal{B}$ is the set of all possible contexts. Solving the conditional MSBP (to be discussed in the upcoming sections) then allows us to obtain the maximum-likelihood distribution of the task's microarchitectural  execution state at \emph{any time}, and for \emph{any} $\beta\in\mathcal{B}$. These distributions can themselves then be sampled, at various times, to obtain `synthetic' execution profiles \emph{for unknown resource allocations}.


\revg{In the following, we present the conditional MSBP, building on the exposition of \cite{bondar-2024-psmsbp} for the \emph{unconditional} MSBP. Sec. \ref{sec:conditionalmsbp} formulates the conditional MSBP, Sec. \ref{subsec:DulaityForMSBP} presents its solution, and Sec. \ref{subsec:FromConditionalMSBPtoGenerativeProfiling} illustrates an application of generative profiling to multicore resource allocation.}

\subsection{Conditional MSBP}\label{sec:conditionalmsbp}
\noindent For any time $t$, define the augmented state 
$$\eta(t) := \begin{pmatrix}
\xi(t)\\
\beta
\end{pmatrix}\in\mathcal{X}\times\mathcal{B}\subset\mathbb{R}^{m+b},$$
and let $\eta(t)\sim\mu_{t}$ where $\sim$ denotes ``follows the distribution". By Bayes' theorem \cite[p. 169]{bertsekas2008introduction}, we have
\begin{align}
\xi(t) \mid \beta \sim \frac{\mu_t}{\int_\mathcal{X} \mu_t \:\differential\xi} \:.
\label{ComputeConditional}
\end{align} 
The numerator in \eqref{ComputeConditional} is the joint distribution of the augmented state $\eta(t)$, and the denominator is the marginal distribution of the resource budget $\beta\in\mathcal{B}\subset\mathbb{R}^{b}$.

Below, we will detail how to learn the numerator $\mu_t$ in \eqref{ComputeConditional} via the MSBP. Since the denominator in \eqref{ComputeConditional} is independent of time $t$, it can be pre-computed before solving the MSBP.

Given a set of $n_{b}$ resource contexts $\{\beta^j\}_{j\in\llbracket n_b\rrbracket}$ where for each we have the corresponding $n_d$ profiles $\{\xi^{i,j}(t)\}_{i\in\llbracket n_d\rrbracket}$, we construct the \emph{empirical distributions}
\begin{align}
\mu_\sigma := \frac{1}{n_dn_b} \sum_{i=1}^{n_d} \sum_{j=1}^{n_b}\delta(\eta-\eta^{i,j}(t_\sigma)), \quad \forall\sigma\in\llbracket n_{s}\rrbracket,
\label{defmusigmaempirical}    
\end{align}
supported over the augmented state space $\mathcal{X}\times\mathcal{B}\subset\mathbb{R}^{d+b}$,  where $\delta$ denotes the Dirac delta, i.e., 
$$\delta\left(\eta - \eta_{0}\right) := \begin{cases}
1 & \text{if} \quad \eta = \eta_0,\\
0 & \text{otherwise.}
\end{cases}
$$
Thus, each distribution $\mu_\sigma$ consists of scattered data points at $\eta^{i,j}$, representing samples of microarchitectural execution state $\xi^{i,j}$ conditioned on the resource context $\beta^{j}$. The $n_{s}$ measurement instances $t_{\sigma\in\llbracket n_s\rrbracket}$ are the times
\[ 0=t_1<t_2<\dots<t_{n_s-1}<t_{n_s},\]
at which `snapshots' of $\xi$ were taken for each of the $n_d$ profiles.

Given the $n_s$ empirical distributions thus constructed, our goal is to find the most likely measure-valued path $t\mapsto\mu_t$, where $\eta(t)\sim\mu_t$ $\forall t\in[t_1,t_{n_s}]$,
satisfying the distributional constraints 
\begin{align}
\eta(t_\sigma)\sim\mu_\sigma \qquad\forall\sigma\in\llbracket n_s\rrbracket.
\label{DistributionalConstraints}
\end{align}
Therewith, we can sample the distributions $\mu_t$ along this path with arbitrary temporal granularity to obtain our synthetic profiles. The path $t\to\mu_t$ itself can be obtained from the multimarginal Schr\"odinger bridge (MSB) between our known distributions $\{\mu_\sigma\}$. MSBP enables this by finding the maximum-likelihood probability mass transport plan (the MSB) -- the solution of an MSBP -- between the known distributions $\{\mu_\sigma\}$. 

\noindent{\bf Mass Transport Plan ($\bm{M}$).} A \emph{transport plan} between the distributions $\{\mu_\sigma\}$ is an order-$n_s$ tensor $\bm{M}\in(\mathbb{R}^{n_{d}n_{b}})_{\geq 0}^{\otimes n_s}$, where $[\bm{M}_{i_1,\dots,i_{n_s}}]$ is the probability mass transported between the points $\eta^{i_1}(t_1),\dots,\eta^{i_{n_s}}(t_{n_s})$, and where $\otimes$ denotes the tensor product. Such a plan must also `conserve' probability mass, i.e., for all $\sigma\in\llbracket n_s\rrbracket$ the mass transported to/from all points in the support of $\mu_\sigma$ must equal to $\mu_\sigma$. Concretely, let
\begin{equation}
    \Big[{\rm{proj}}_\sigma(\bm{M})_j\Big] := \!\!\!\!\!\sum_{i_1,\dots,i_{\sigma-1},i_{\sigma+1},\dots,i_{n_s}}\!\!\!\!\!\!\!\!\!\bm{M}_{i_1,\dots,i_{\sigma-1},j,i_{\sigma+1},\dots,i_{n_s}}.\label{DefUnimargProj}
\end{equation}
This constraint can then be expressed as
\begin{equation}
{\rm{proj}}_\sigma(\bm{M})=\mu_\sigma \quad\quad\forall\sigma\in\llbracket n_s\rrbracket,
\label{MSBP_constraint}
\end{equation}
wherein the mapping \[{\rm{proj}}_\sigma:(\mathbb{R}^{n_{d} n_{b}})_{\geq 0}^{\otimes n_s}\to\mathbb{R}_{\geq 0}^{n_{d}n_{b}}.\] is called the \emph{unimarginal projection} of $\bm{M}$ onto $\mu_\sigma$. Similarly, we may define the \emph{bimarginal projection} of $\bm{M}$ onto its $(\sigma_1,\sigma_2)$th components as the mapping 
\[{\rm{proj}}_{\sigma_1,\sigma_2}:(\mathbb{R}^{n_{d}n_{b}})_{\geq 0}^{\otimes n_s}\to\mathbb{R}_{\geq 0}^{n_{d}n_{b}\times n_{d}n_{b}},\]
where
\begin{align}
   & \Big[{\rm{proj}}_{\sigma_1,\sigma_2}(\bm{M})_{j,\ell}\Big] := \nonumber\\ 
   & \sum_{\{i_\sigma\}_{\sigma\in\llbracket n_s\rrbracket\setminus\{\sigma_1,\sigma_2\}}}\!\!\!\!\!\!\!\!\!\!\!\!\!\!\bm{M}_{i_1,\dots,i_{\sigma_1-1},j,i_{\sigma_1+1},\dots,i_{\sigma_2-1},\ell,i_{\sigma_2+1},\dots,i_{n_s}}\:. \label{DefBimargProj}
\end{align}
Intuitively, ${\rm{proj}}_{\sigma_1,\sigma_2}(\bm{M})$ is a restriction of $\bm{M}$ to two distributions $(\mu_{\sigma_1},\mu_{\sigma_2})$. This projection may also be viewed as a scattered distribution supported over the finite set $\{\eta^i(t_{\sigma_1})\}_{i\in\llbracket n_{d}n_{b}\rrbracket}\times\{\eta^i(t_{\sigma_2})\}_{i\in\llbracket n_{d}n_{b}\rrbracket}$ with $\mu_{\sigma_1}$ and $\mu_{\sigma_2}$ as statistical marginals.

\begin{figure}[t]
    \centering
    \includegraphics[width=.95\linewidth]{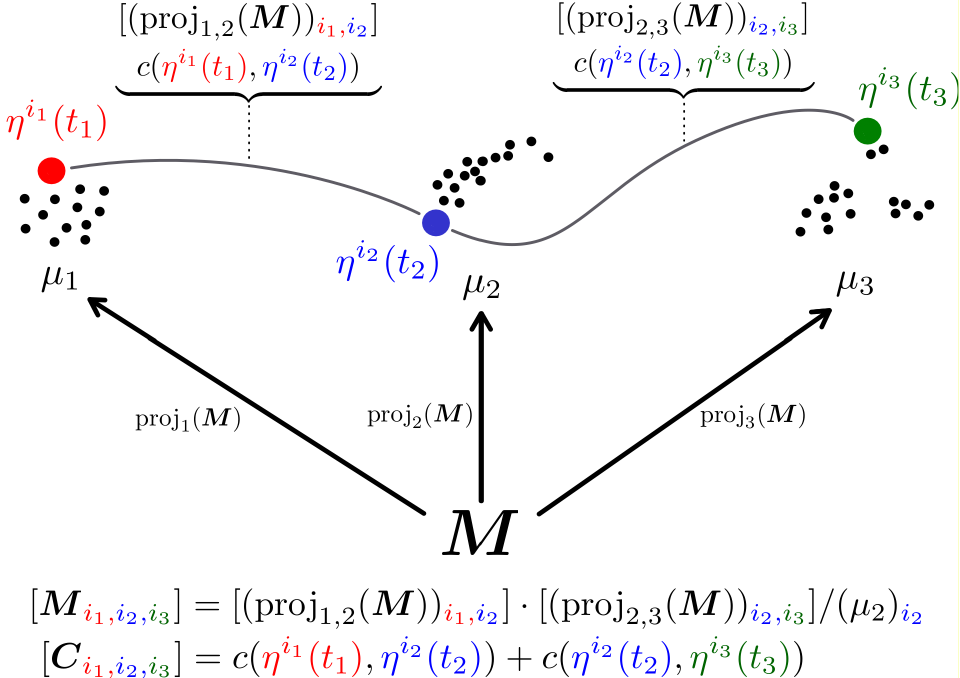}\vspace{-.5ex}
    \caption{The relationship between a transport plan $\bm{M}$, its projections, and the transport cost $\bm{C}$ for $n_s=3$ empirical distributions $\{\mu_\sigma\}_{\sigma\in\llbracket n_s\rrbracket}$. Unimarginal projections equal the distributions, and for the three arbitrary colored points, $[{\bm{M}}_{i_1,i_2,i_3}]$ encodes the total probability mass transported, and $[{\bm{C}}_{i_1,i_2,i_3}]$ the per-unit cost thereof, along the grey-colored path therebetween. Similarly, the bimarginal projections of $\bm{M}$ and the pairwise cost $c$ pertain to the pairwise sections of the path.}
    \vspace{-3ex}
    \label{fig:MSBP_overview}
\end{figure}

\noindent{\bf Ground Cost ($\bm{C})$.} As many such transport plans $\bm{M}$ may exist, we introduce a transport cost function or ``ground cost"
$$c:(\mathbb{R}^{n_{d}n_{b}})^{n_s}\to\mathbb{R}_{\geq 0},$$
which is such that $c(\eta^{i_1}(t_1),\dots,\eta^{i_{n_s}}(t_{n_s}))$ is the cost of transporting unit probability mass between the $n_{s}$ tuple $\left(\eta^{i_1}(t_1),\dots,\eta^{i_{n_s}}(t_{n_s})\right)$\footnote{\revg{As opposed to the previously introduced $\eta^{i,j}$ notation, $\eta^i$, single-indexed, is used when context $\beta^j$ is not of relevance, such as when computing cost.}}. Intuitively, this cost encodes the `distance' between the supports of the $\{\mu_\sigma\}$. With this cost function, we may then define $\bm{C}\in(\mathbb{R}^{n_{d}n_{b}})_{\geq 0}^{\otimes n_s}$ as tensor of order $n_{s}$, whose entries are 
\begin{align}
[\bm{C}_{i_1,\dots,i_{n_s}}] = c(\eta^{i_1}(t_1),\dots,\eta^{i_{n_s}}(t_{n_s})),
\label{CostTensorEntriesGeneral}    
\end{align}

Figure \ref{fig:MSBP_overview} shows how a valid transport plan $\bm{M}$ (i.e., one satisfying \eqref{MSBP_constraint}), its projections, and the transport cost $\bm{C}$ relate to the distributions $\{\mu_\sigma\}$. Note that in the context of this work, the snapshots $\mu_\sigma$ are indexed by times $t_\sigma$, which imposes a natural `path' structure to the cost tensor $\bm{C}$, i.e., \eqref{CostTensorEntriesGeneral} specializes to
\begin{align}
[\bm{C}_{i_1,\dots,i_{n_s}}] = \sum_{j=1}^{n_s-1}[C^j_{i_{j},i_{j+1}}],
\label{CostTensorPathStructure}
\end{align}
where the $C^j\in\mathbb{R}^{n_d\times n_d}$ are matrices whose $(i,j)$th entry is the value $c(\eta^{i_j}(t_j),\eta^{i_{j+1}}(t_{j+1}))$ for some cost function $c$ (in this work, we use the squared Euclidean cost $c(\cdot,\cdot):=\|\cdot-\cdot\|_2^2$). Intuitively if $n_s=3$ (see Figure \ref{fig:MSBP_overview}), \eqref{CostTensorPathStructure} states that the cost of transport between $\mu_1,\:\mu_2,$ and $\mu_3$ is equal to the sum of the costs between $(\mu_1$ and $\mu_2)$ and $(\mu_2$ and $\mu_3)$.

The introduction of a ground cost $\bm{C}$ allows for the assignment to any transport plan $\bm{M}$ a mass transport cost $$\langle\bm{C},\bm{M}\rangle:=\sum_{i_{1},\hdots,i_{n_{s}}}[\bm{C}_{i_1,\dots,i_{n_s}}][\bm{M}_{i_1,\dots,i_{n_s}}],$$ where $\langle\cdot,\cdot\rangle$ denotes the Hilbert-Schmidt inner product. Thus, we may compare probability mass transport plans in terms of this transport cost. For instance, if $c(\eta^{i_1}(t_1),\eta^{i_2}(t_2))$ in Figure \ref{fig:MSBP_overview} is relatively large, a plan which transports more probability mass between these two points may have a higher total mass transport cost.

\noindent{\bf Problem Formulation.} The MSBP refers to finding the transport plan $\bm{M}$ with the minimal entropy-regularized probability mass transport cost, expressed as the optimization problem
\begin{subequations}\label{DiscreteMSBP}
\begin{align}
& \bm{M}_{\rm{opt}}:=\underset{\bm{M}\in\left(\mathbb{R}^{n_{d}n_{b}}\right)^{\otimes n_s}_{\geq 0}}{\arg\min}~\langle\bm{C}+\varepsilon\log\bm{M},\bm{M}\rangle \label{DiscreteMSBPobj} \\
& \text{subject to}~~{\rm{proj}}_{\sigma}\left(\bm{M}\right) = \bm{\mu}_{\sigma}\quad\forall\sigma\in\llbracket n_s\rrbracket, \label{DiscereteMSBPconstr}
\end{align}
\end{subequations}
where $\varepsilon>0$ is the entropy regularization\footnote{If $\varepsilon=0$, then the objective \eqref{DiscreteMSBPobj} becomes the probability mass transport cost $\langle\bm{C},\bm{M}\rangle$, and \eqref{DiscreteMSBP} becomes the multimarginal Monge-Kantorovich optimal transport problem \cite{pass2015multi}.} parameter. We refer to the minimizer of \eqref{DiscreteMSBP} as the Multimarginal Schr\"odinger Bridge (MSB) and denote it as $\bm{M}_{\rm{opt}}$. 

What makes the MSBP \eqref{DiscreteMSBP} pertinent to learning from distributional data is a mathematical result from the theory of large deviations \cite{dembo2009large}, specifically Sanov's theorem \cite{sanov1958probability}, \cite[Sec. II]{follmer1988random}. This theorem establishes that the minimizer of \eqref{DiscreteMSBP} is precisely the \emph{most likely} joint distribution subject to the observational constraints \eqref{DiscereteMSBPconstr} imposed at times where measurements are available. In other words, solving \eqref{DiscreteMSBP} is equivalent to solving a \emph{constrained maximum likelihood problem on the space of probability measure-valued curves $t\mapsto\mu_{t}$ traced over time}. \reva{Mathematical details for this maximum likelihood guarantee are provided in Supplementary Section S1.
}

\subsection{Leveraging Duality to Solve Conditional MSBP}\label{subsec:DulaityForMSBP}
\noindent The MSBP \eqref{DiscreteMSBP} is a \emph{strictly convex} program in $\left(n_{d}n_{b}\right)^{n_s}$ decision variables. Thus, the existence-uniqueness of the minimizer $\bm{M}_{\rm{opt}}$ is guaranteed. Using strong Lagrange duality \cite[Ch. 5.2.3]{boyd2004convex}, it follows that the MSB admits the structure 
\begin{align}
\bm{M}_{\rm{opt}}=\bm{K}\odot\bm{U},\quad \bm{K}=\exp\bigg(\frac{-\bm{C}}{\varepsilon}\bigg),\quad \bm{U}\!\!=\!\!\!\!\bigotimes_{\sigma\in\llbracket n_s\rrbracket}\!\!u_\sigma,
\label{Moptstructure}
\end{align}
wherein $\odot$ denotes the elementwise product, and the vectors
\[u_\sigma:=\exp(\lambda_\sigma/\varepsilon)\in\mathbb{R}^{n_{d}n_{b}}_{>0}\]
are the exponential-transformed Lagrange multipliers $\lambda_\sigma$ associated with the equality constraints \eqref{DiscereteMSBPconstr}. Note that the tensors $\bm{K},\bm{U}\in(\mathbb{R}^{n_{d}n_{b}})_{\geq 0}^{\otimes n_s}$. 

Thanks to this structure \eqref{Moptstructure} for the optimal solution, we only need to determine the $\{u_\sigma\}_{\sigma\in\llbracket n_s\rrbracket}$, which are obtained by the Sinkhorn iterative scheme \cite{sinkhorn1964relationship,sinkhorn1967concerning}
\begin{equation}\label{eq:MultimarginalSink}   
{u}_{\sigma} \leftarrow {u}_{\sigma} \odot {\mu}_{\sigma}\oslash{\rm{proj}}_{\sigma}\left(\bm{K}\odot\bm{U}\right) \quad\quad\forall\sigma\in\llbracket n_s\rrbracket . 
\end{equation}
In \eqref{eq:MultimarginalSink}, the symbol $\oslash$ denotes elementwise division. The Sinkhorn recursions \eqref{eq:MultimarginalSink} are strictly contractive over the positive orthant $\mathbb{R}^{n_{d}n_{b}}_{>0}$ with respect to Hilbert's projective metric \cite{birkhoff1957extensions,bushell1973hilbert,kohlberg1982contraction}
\begin{align}
{\texttt{Hilb}}(p,q):=\log \left(\frac{\max _{i=1, \ldots, n} p_i / q_i}{\min _{i=1, \ldots, n} p_i / q_i}\right) \,\forall p, q\in\mathbb{R}^{n}_{>0}.
\label{defHilbertMetric}    
\end{align}
Therefore, by the Banach contraction mapping theorem \cite{banach1922operations}, the recursions are guaranteed to converge to a unique fixed point with linear rate of convergence.

\begin{algorithm}[t!]
\caption{Generative profiles via conditional MSBP}
\begin{algorithmic}[1]
    \Require Empirical profiles $\{\xi^{i,j}(t_{\sigma})\}$ for $\beta^{j}\in\mathcal{B}'$, $(i,j)\in\llbracket n_{d}\rrbracket\times\llbracket n_{b}\rrbracket$, where snapshot index $\sigma\in\llbracket n_{s}\rrbracket$, \revg{profile granularity $\Delta t$}, entropic regularization parameter $\varepsilon>0$
    \vspace*{0.1in}
    \State $\eta^{i,j}(t_{\sigma}) \gets \!\!\begin{pmatrix}\xi^{i,j}(t_{\sigma})\\ \beta^{j}\end{pmatrix}\!\!\in\mathbb{R}^{m+b}$ \Comment{augmented state samples}
    \State Construct $\{\mu_\sigma\}_{\sigma\in\llbracket n_s\rrbracket}$ using \eqref{defmusigmaempirical} \Comment{empirical distributions}
    \vspace*{-0.1in}
    \NoNumber{\LeftComment{construct $n_{s}-1$ squared Euclidean distance matrices}}
    \For{$j \in 1:n_{s}-1$}
        \State $[C^j_{i_{j},i_{j+1}}]\gets \|\eta^{i_j}(t_j)-\eta^{i_{j+1}}(t_{j+1})\|_2^2$
    \EndFor\\
    
    \State Construct $\bm{C}$ using \eqref{CostTensorPathStructure} \Comment{cost tensor}

    \State $\{u_\sigma\}_{\sigma\in\llbracket n_s\rrbracket} \gets {\texttt{MSBP\_Solve}}({\bm{C}},\{\mu_\sigma\}_{\sigma\in\llbracket n_s\rrbracket},\varepsilon)$\\
    
    \For{$t \in \revg{t_1:\Delta t: t_{n_s}}$}
        \State Compute $\mu_t$ via \eqref{BimargPostprocessing}-\eqref{eq:mu_interpolation} \Comment{predict joint distribution}
        \For{$\beta\in\mathcal{B}$}
            \State $\nu_t \gets \mu_{t}/\left(\frac{1}{n_{b}}\sum_{j=1}^{n_{b}}\delta(\beta-\beta^{j})\right)$ \Comment{predict per 
            \eqref{ComputeConditional}} 
            
            \State $\xi(t) \mid \beta \gets \texttt{maxlikelihood}(\nu_t)$
        \EndFor
    \EndFor
\end{algorithmic}
{\bf{Result:}} Synthetic profiles $\{\xi(t)\mid\beta\}_{t\in\revg{t_1:\Delta t: t_{n_s}}}$ for all $\beta\in\mathcal{B}$
\label{alg:comp_framework}
\end{algorithm}
\begin{algorithm}[t!]
\caption{MSBP Solver $\texttt{MSBP\_Solve}$}
\begin{algorithmic}[1]
    \Require Cost tensor $\bm{C}$, probability distributions $\{\mu_\sigma\}_{\sigma\in\llbracket n_s\rrbracket}$ entropic regularization parameter $\varepsilon>0$, random vector generator $\texttt{rand}$, numerical tolerance $\texttt{tol}$, maximum number of iterations $\texttt{maxiter}$
    \State $\bm{K}\gets\exp\left(-\bm{C}/\varepsilon\right)$ \Comment{scaled elementwise exponential}
    \State $u_{\sigma}(:,1)\gets\texttt{rand}_{n_{d}n_{b}\times 1}\;\forall\sigma\in\llbracket n_s\rrbracket$ \Comment{random init.}
    \State $\texttt{idx} \gets 1$ \Comment{ initialize Sinkhorn recursion index}
    \State ${\texttt{err}}_{\sigma}(\texttt{idx})\gets 1_{n_{d}n_{b}\times 1}\;\forall\sigma\in\llbracket n_s\rrbracket$ \Comment{initialize errors}
    \While{($\max_{\sigma\in\llbracket n_s\rrbracket} \texttt{err}_{\sigma}(\texttt{idx})> \texttt{tol}$) \textbf{and} $({\texttt{idx}} <  {\texttt{maxiter}})$}
    \LeftComment{Sinkhorn recursion \eqref{eq:MultimarginalSink}}
        \For{$\sigma\in\llbracket n_{s}\rrbracket$}
            \State $u_\sigma(:,\texttt{idx}+1) \leftarrow {u}_{\sigma}(:,\texttt{idx})\odot\mu_\sigma\oslash{\rm{proj}}_{\sigma}\!\!\left(\bm{K}\odot\bm{U}\right)$
        \State \parbox[t]{\dimexpr\linewidth-\algorithmicindent}{   ${\texttt{err}}_{\sigma}(\texttt{idx}+1)\gets{\texttt{Hilb}}(u_\sigma(:,\texttt{idx}),$  \\
        \hspace*{2em}$ u_\sigma(:,\!\texttt{idx}\!+\!1))$ \Comment{error in Hilbert metric~\eqref{defHilbertMetric} } }
        \EndFor
    \State $\texttt{idx}\gets \texttt{idx}+1$    
    \EndWhile
\end{algorithmic}
{\bf{Result:}} Converged $\{u_\sigma\}_{\sigma\in\llbracket n_s\rrbracket}$ defining the MSB ${\bm{M}}_{\rm{opt}}$
\label{alg:msbp_solver}
\end{algorithm}

Notice also that both the formulation and the solution of the MSBP \eqref{DiscreteMSBP} are nonparametric in the sense that no statistical parametrization (e.g., mixture of Gaussian, exponential family, knowledge of moment or sufficient statistic) is assumed whatsoever on the joint $\bm{M}_{\rm{opt}}$.

\subsection{From Conditional MSBP to Generative Profiling}\label{subsec:FromConditionalMSBPtoGenerativeProfiling}
\noindent\textbf{Maximum likelihood distributions.}
Once $\bm{M}_{\rm{opt}}$ is computed using \eqref{Moptstructure}-\eqref{eq:MultimarginalSink}, we obtain the distributions $\mu_t$ in \eqref{ComputeConditional} as follows. Let
\begin{align}
M_{\rm{opt}}^\sigma := {\rm{proj}}_{\sigma,\sigma+1}(\bm{M}_{\rm{opt}}).
\label{BimargPostprocessing}    
\end{align}
For $t\in[t_1,t_{n_s}]$, using \eqref{BimargPostprocessing}, we compute
\begin{equation}
    \mu_t = \sum_{i=1}^{n_{d}n_{b}}\sum_{j=1}^{n_{d}n_{b}}[(M_{\rm{opt}}^\sigma)_{i,j}]\delta\Big(\eta-\bigr((1-\lambda)\eta^i(t_\sigma)+\lambda\eta^j(t_{\sigma+1})\bigr)\Big)
\label{eq:mu_interpolation}
\end{equation}
where $t\in[t_\sigma,t_{\sigma+1}]$ and $\lambda:=\tfrac{t-t_\sigma}{t_{\sigma+1}-t_\sigma}\in[0,1]$. Each such $\mu_t$ is a weighted scattered distribution of $n_{d}^{2}n_{b}^2$ data points, and represents the most likely distribution for $\eta(t)$ such that $\eta(t_\sigma)\sim\mu_\sigma$ for all $\sigma\in\llbracket n_s\rrbracket$ and given transport cost $\bm{C}$. \revg{In other words, $\eta(t)$ is the least assumptive model which explains the data ($\mu_\sigma$ and $\bm{C}$, in this case). We provide more details on the maximum-likelihood guarantee in Supp. Sec. S1.
}

Notice that the unimarginal projection \eqref{DefUnimargProj} is needed for the Sinkhorn recursion \eqref{eq:MultimarginalSink}, while the bimarginal projection \eqref{DefBimargProj} is needed for the post-processing \eqref{BimargPostprocessing}. {Efficient computation of these projections is detailed in Supp. Sec. S2. 
}

\noindent\textbf{Computational complexity.} The complexity of the Sinkhorn recursion \eqref{eq:MultimarginalSink} is governed by that of the unimarginal projections \eqref{DefUnimargProj} which, in general, is exponential in $n_s$. As shown in previous work (\cite{bondar-2024-psmsbp,bondar-2024-stochastic}), this complexity becomes \emph{linear} in $n_s$ when we take into account the temporal path structure \eqref{CostTensorPathStructure} of our cost (see Supp. Sec. S2
for details).

Thus, we may solve \eqref{DiscreteMSBP} with computational complexity $\mathcal{O}\left((n_{s} - 1)n_{d}^{2}n_{b}^{2}\right)$. This linear complexity with respect to $n_{s}$ is particularly significant considering that the primal formulation \eqref{DiscreteMSBP} has an exponential complexity $\mathcal{O}\left(\left(n_{d}n_{b}\right)^{n_{s}}\right)$.

\noindent\textbf{Maximum likelihood profiles.} Having computed both the numerators and denominators in \eqref{ComputeConditional}, we obtain the conditional distributions of the random vector $\xi(t)\mid\beta$.
By re-sampling the computed conditional distributions \eqref{ComputeConditional}, we are able to generate high conditional probability samples, and therefore generative or synthetic profiles. This re-sampling can, for instance, be done by simply returning the top few high probability samples since the distribution values evaluated at the samples, and not just the samples of \eqref{ComputeConditional}, are available from MSB computation. Alternatively, this re-sampling can be done using existing diffusion models \cite{sohl2015deep,song2019generative,song2020score} for generative AI.

The overall computational framework discussed above is summarized in Algorithm \ref{alg:comp_framework}. First, we form the distributions $\{\mu_\sigma\}$ from our empirical profiles for the known context set $\mathcal{B}'$ (lines 1-2). Second, we construct the cost tensor $\bm{C}$ from the squared-Euclidean ground costs between known support points $\eta$ (lines 3-7). On line 8, we invoke an MSBP solver (details in Algorithm \ref{alg:msbp_solver}) to obtain the vectors $\{u_\sigma\}$, which may be used as in \eqref{Moptstructure} to obtain the MSB ${\bm{M}}_{\rm{opt}}$. Finally in lines 10-16 we create the synthetic profiles -- for each time $t$ of interest (e.g., small timesteps \revg{$\Delta t$} from \revg{initial time $t_1$} to $t_{n_s}$) use the MSB to obtain the joint distribution $\mu_t$ of $\eta(t)$ (line 11), then therefrom for all $\beta\in\mathcal{B}$ the conditional distribution $\nu_t$ of $\xi(t)\mid\beta$ (line 13), which we sample in line 14 to obtain the execution state $\xi(t)\mid\beta$.
\begin{remark}
Algorithm \ref{alg:comp_framework} returns the \emph{maximum-likelihood} synthetic profiles for the given task. The \texttt{maxlikelihood} oracle in line 14 of the algorithm can be modified to produce, for example, \emph{mean} synthetic profiles by extracting the mean of the conditional distribution. {Additionally, a random sampling scheme may be used to generate multiple high-probability random synthetic profiles.}
\end{remark}

\noindent\textbf{The MSBP Solver.} Algorithm \ref{alg:comp_framework} makes use of an MSBP solver, which we detail in Algorithm \ref{alg:msbp_solver}. The solver forms $\bm{K}$ from $\bm{C}$ and initial random guesses for the vectors $\{u_\sigma\}$ (lines 1-2). The loop on lines 5-11 performs the Sinkhorn iterations \eqref{eq:MultimarginalSink} until all vectors $u_\sigma$ converge in the Hilbert metric. These vectors are then returned as the solution (recall that by \eqref{Moptstructure}, the MSB ${\bm{M}}_{\rm{opt}}$ is thereby uniquely defined). Note that line 7 requires the unimarginal projection to be computed -- this is the greatest computational bottleneck of the algorithm, and also where the  structure of $\bm{C}$ is often exploited to greatly accelerate the computation.

\revg{
\noindent\textbf{Computational Considerations.}
Practical implementations of Algorithms \ref{alg:comp_framework} and \ref{alg:msbp_solver} benefit from considering the computational stability thereof. In MSBP solvers (Algorithm \ref{alg:msbp_solver}), the parameter $\varepsilon$ is generally chosen to be a small, positive value; larger values lead to faster convergence of the solver at the cost of noisier solutions.
The dual variables $u_\sigma$ may be initialized to any positive values. In practice a random initialization within $(0,1)$ is used. Similarly, any positive cost tensor $\bm{C}$ may be used -- we use the Euclidean cost (line 4, Algorithm \ref{alg:comp_framework}), as is common throughout the related literature \cite{bondar-2024-psmsbp,bondar-2024-stochastic}. Since components of $\eta$ tend to vary greatly in scale, however, for each snapshot time $t_\sigma$ we scale the $\eta$ componentwise into the range $[0,0.1]$. The scaling is done so that the Euclidean distance takes into equal account all components thereof, and the range is chosen for stability of Algorithm \ref{alg:msbp_solver} (large values of $\eta$ $\rightarrow$ large values in $\bm{C}$ $\rightarrow$ near-zero values in $\bm{K}$, which can lead to underflow). 
Finally, $\mu_t$ as computed in line 11 of Algorithm \ref{alg:comp_framework} 
implicitly needs to have support covering all of $\mathcal{B}$ for line 13 to be feasible. To ensure this, we employ a standard KDE alignment of $\mu_t$ as computed via \eqref{BimargPostprocessing}-\eqref{eq:mu_interpolation}. Other domain alignment approaches may be used as well.
}

%% file: accuracy-eval.tex
\section{Effectiveness of generative profiling}\label{sec:accuracy-eval}

\subsection{Benchmarks, Hardware, and Workload Measurement}\label{sec:workload-measurement}

\noindent{\bf Benchmarks.} 
To evaluate generative profiling, we used the PARSEC~\cite{PARSEC08} benchmark suite, which includes programs with diverse execution characteristics. This suite has been widely used in prior work on  scheduling and resource allocation (e.g.,~\cite{dna-rtas21,Kim16,Yun16-memguard-journal}).  We ran these benchmarks in single-threaded mode, with the \texttt{smallsim} input. \reva{The observed execution times ranged from approximately 100 milliseconds to several seconds on our experimental platform.}

\noindent{\bf Hardware.} 
We performed measurements on an Intel Xeon E5-2618L v3 processor with 8 cores, 20MB 20-way set-associative L3 cache, and a single channel 8GB PC-2133 DDR4 DRAM. The machine supports both Cache Allocation Technology (CAT)~\cite{intelCAT} for cache allocation and DVFS for power management. CAT divides the shared L3 cache into $N_{\mathsf{ca}} = 20$ equal-size partitions. Using the method in~\cite{Yun16-memguard-journal}, we measured a maximum guaranteed bandwidth of 1.4 GB/s on each machine, which we divided into  $N_{\mathsf{bw}} = 20$ partitions of 70MB/s each using MemGuard~\cite{Yun16-memguard-journal}. We considered  $N_{\mathsf{freq}} = 12$ CPU frequency settings, ranging from 1.2GHz to 2.3GHz, at steps of 0.1GHz. Since the machine requires a minimum of 2 cache partitions per allocation, it supports 
$(N_{\mathsf{ca}}-1) \times N_{\mathsf{bw}} \times N_{\mathsf{freq}} = 19 \times 20 \times 12 = 4560$ different resource contexts in total. To ensure deterministic timing, we disabled cache prefetching and CPU hyperthreading via the system's BIOS. 

\noindent{\bf Measurements.} To gather the empirical profiles used as input to our model, we ran each workload on a dedicated core under the desired resource context, and collected data for 100 runs. In each run, we used the CPU’s performance counters to periodically record the number of instructions retired, cache requests, and cache misses ($\xi_1, \xi_2, \xi_3$, respectively in our model), once every 10 ms. 

Our  model requires only a small subset of empirical profiles for a small number of resource contexts as training inputs. However, to evaluate its accuracy, we also needed the ``ground-truth'' data. To obtain the ground-truth profiles, we performed measurements for all possible resource contexts for each benchmark program. The whole measurement process took us over two months to complete. Since we collected 100 profiles per context, the ground-truth data contain $456,\!000$ empirical profiles per benchmark.

\begin{figure*}[t!]
    \centering
    \begin{subfigure}[b]{0.490\textwidth}
        \centering
        \includegraphics[width=\textwidth]{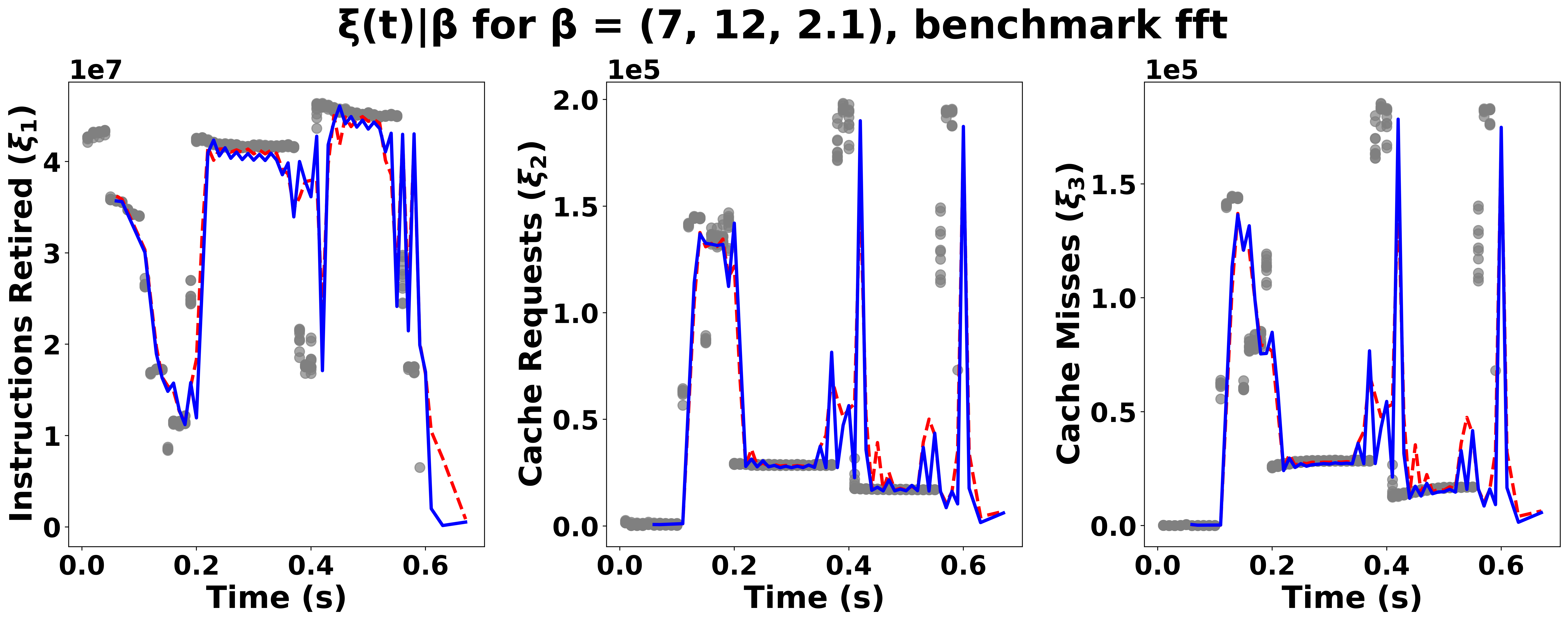}
    \end{subfigure}
    \hfill
    \begin{subfigure}[b]{0.490\textwidth}  
        \centering 
        \includegraphics[width=\textwidth]{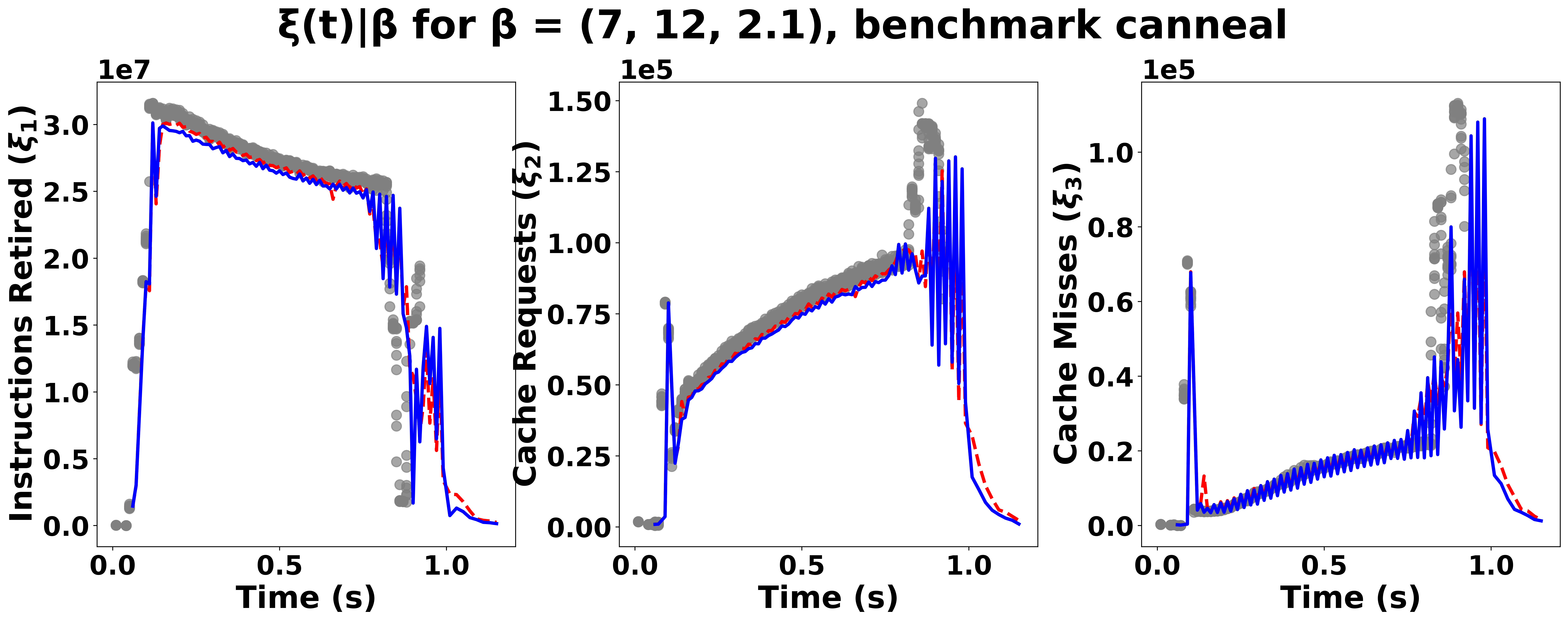}
    \end{subfigure}
    \caption{{{Maximum-likelihood synthetic profile ({\color{blue}blue}), mean synthetic profile ({\color{red}red}), and all empirical profiles ({\color{gray}grey}) for the benchmarks $\mathsf{fft}$ and $\mathsf{canneal}$ when $\beta=(7,\:12,\:2.1)^\top$}} on our default experimental platform.} 
    \vspace{-2ex}
    \label{fig:synth_profiles_accuracy}
\end{figure*}

\begin{figure*}[t!]
    \centering
    \begin{subfigure}[b]{0.490\textwidth}
        \centering
        \includegraphics[width=\textwidth]{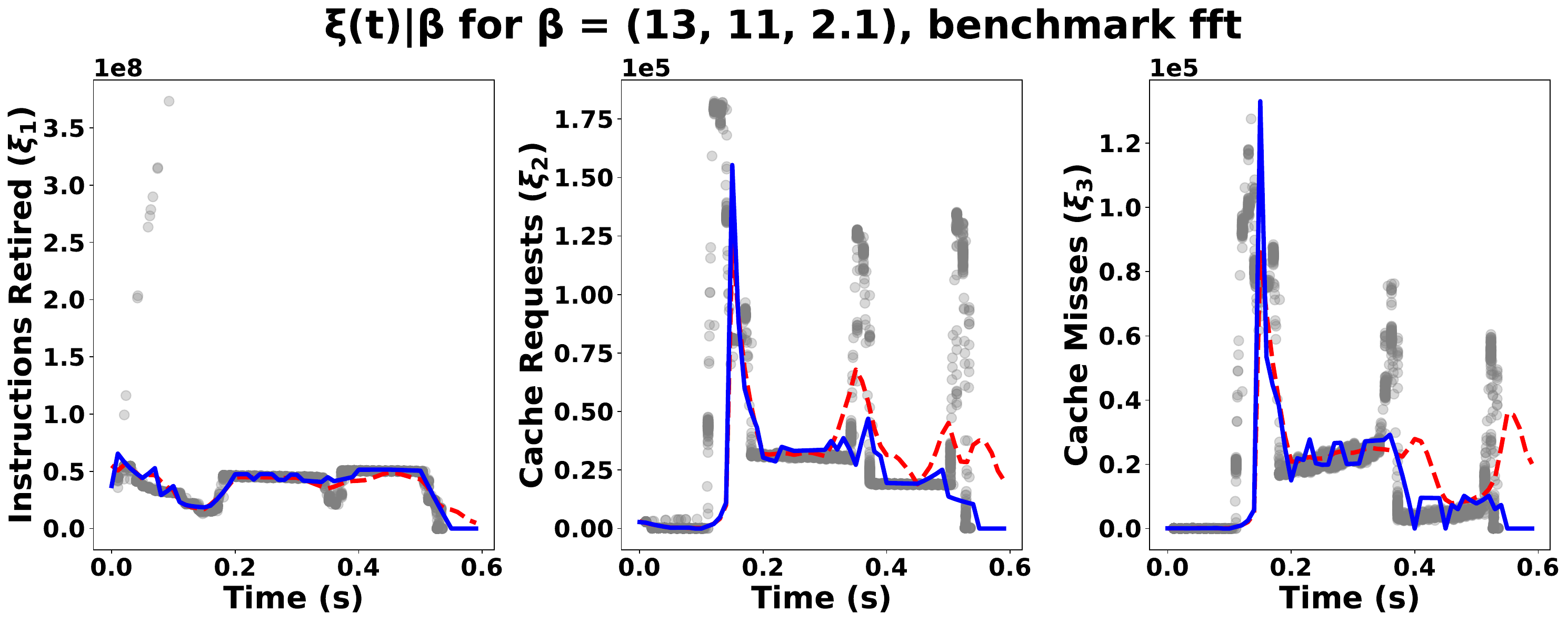}
    \end{subfigure}
    \hfill
    \begin{subfigure}[b]{0.490\textwidth}  
        \centering 
        \includegraphics[width=\textwidth]{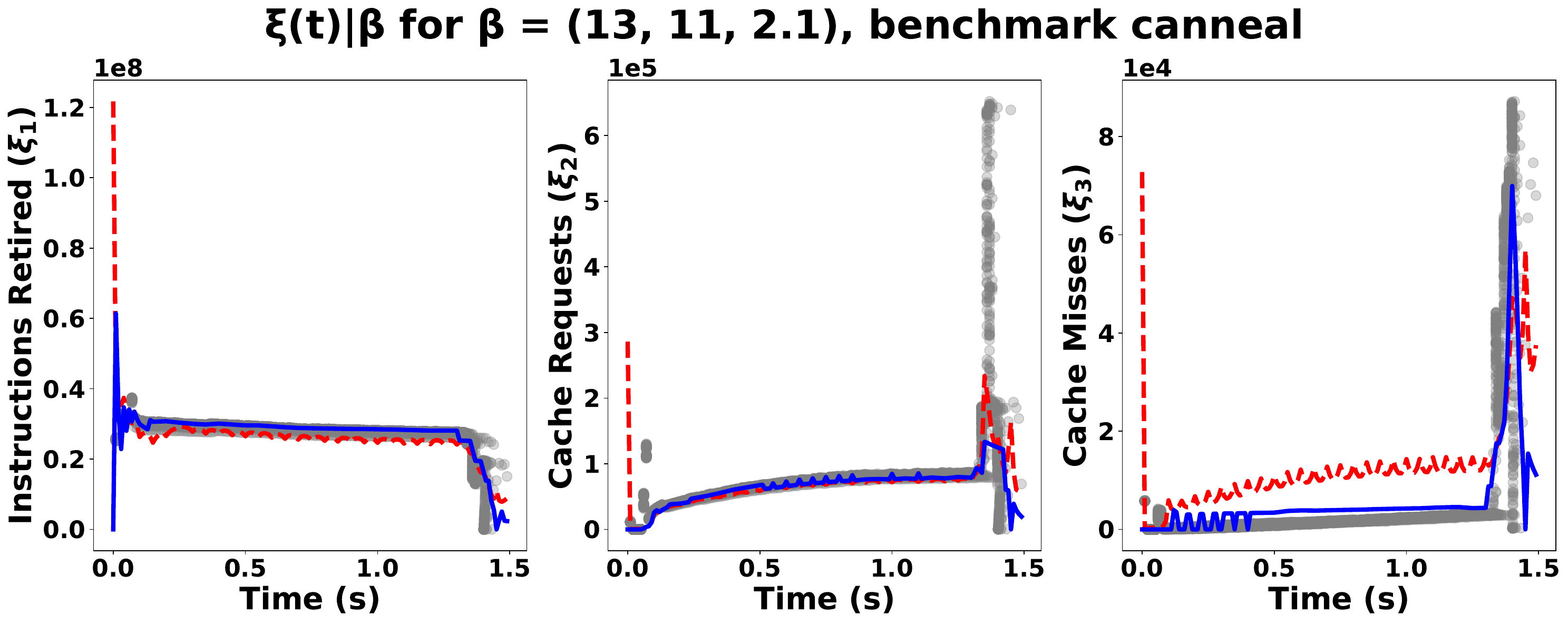}
    \end{subfigure}
    \caption{{{Maximum-likelihood synthetic profile ({\color{blue}blue}), mean synthetic profile ({\color{red}red}), and all empirical profiles ({\color{gray}grey}) for the benchmarks $\mathsf{fft}$ and $\mathsf{canneal}$ when $\beta=(13,\:11,\:2.1)^\top$ on $\mathsf{PlatformLarge}$}}.}
    \vspace{-2ex}
    \label{fig:synth_profiles_platformlarge}
\end{figure*}

\subsection{Accuracy of Generative Stochastic Model}\label{sec:accuracy}
\noindent{\bf Methodology.} Given the empirical profiles, we used the approach detailed in Section~\ref{sec:generative} to generate synthetic profiles for $\xi(t)\mid\beta$. For the experiment, we restricted the number of resource contexts $\beta$ and the number of profiles $n_d$ available to our MSBP solver. Specifically, for
\[ \mathcal{B} =  \{2,3,\dots,20\} \times \llbracket N_{\mathsf{bw}}\rrbracket \times \{1.2, 1.3, \dots, 2.3\} \]
we took a subset $\mathcal{B}'$ of $\mathcal{B}$ as training data. 


We then let $n_d=10$, i.e., we took 10 random profiles from the 100 available for each unique $\beta \in \mathcal{B}'$. From these limited profiles, we constructed the marginal distributions \eqref{defmusigmaempirical} for times $t_\sigma=0.05\cdot(\sigma-1)$, with $t_{n_s}$ chosen depending on the runtime of the benchmark. Thus, with $n_b=|\mathcal{B}'|=125$, each snapshot $\mu_\sigma$ had $n_dn_b=1250$ scattered data points.
For all experiments, our MSBP solver was configured with $\varepsilon=0.1$, $\texttt{tol}=1e\!-\!\!12$, and $\texttt{maxiter}=1e4$. Our dual variables $u_\sigma$ were initialized by $\texttt{rand}=$ uniform samples from $(0,1)$.

\begin{remark}
    The restriction of the number of profiles $(n_d)$ and the number of resource contexts $(n_b)$ provided as input to the MSBP solver has a practical motivation -- if either the number of tasks to profile or the number of possible resource contexts $|\mathcal{B}|$ is large, it is intensive both in time and computation to empirically generate a large number of profiles for all workloads and for each resource context. Our method allows for a significant reduction in profiling time as it requires only a small subset of $\mathcal{B}$ and a small $n_d$ to generate synthetic profiles for all $\beta\in\mathcal{B}$. For instance, for each benchmark task, only $n_dn_b=1250$ empirical profiles with data observations every 50 ms were sufficient to achieve an accuracy near the ground truth consisting of $456,\!000$ empirical profiles with data gathered every 10 ms (i.e., only 0.27\% of the original empirical data was used by our model).
    Moreover, running MATLAB R2024b on a machine with an AMD Ryzen 7 5800X CPU and 80GB of memory, the entire process of MSBP solution and synthetic profile generation for all $\beta\in\mathcal{B}$ took approximately {\bf 15 minutes} per benchmark task. In contrast, collecting the respective empirical profiles for all possible $\beta\in\mathcal{B}$ took many days for each benchmark task (e.g., approximately 2 weeks for the $\mathsf{canneal}$ benchmark).  
\end{remark}

Solving the MSBP over the marginals $\{\mu_{\sigma}\}_{\sigma\in\llbracket n_s\rrbracket}$, we generated the most likely conditional joint distributions $\mu_t/\int_\mathcal{X}\mu_td\xi$ per \eqref{ComputeConditional}, at times from 0 to $t_{n_s}$ every 10 ms, and for each $\beta\in\mathcal{B}$. As each of these joints are \emph{weighted} scattered distributions, we obtained from their aggregate a synthetic profile of the benchmark conditioned on $\beta$ by taking $\xi(t)\mid\beta$ to be the data point with the highest probability value in the respective conditional joint distribution, for each $t\in\{0,0.01,0.02,\dots,t_{n_s}\}$. We call this the \emph{maximum-likelihood} generative profile.

\reva{\noindent{\bf Generative profiles visualized.}  Fig.~\ref{fig:synth_profiles_accuracy} shows the mean and maximum likelihood generative profiles overlaid on 100 empirical profiles for two example benchmarks ($\mathsf{fft}$ and $\mathsf{canneal}$) under a resource context that was farthest from those in the training set, $\beta=\left(7,12,2.1\right)^\top$.
We present these sample results visually to highlight the fine-grained, time-dependent predictive accuracy of our method. Notice that the generative profiles closely align with the empirical profiles, even for  contexts in which $\beta$ is not available to the MSBP solver.}

\noindent{\bf Generality across platforms.} 
To illustrate the generality of our generative profiling, we also evaluated its performance on another, more powerful, multicore platform with more cores and larger cache and memory bandwidth sizes. This platform, referred to as $\mathsf{Platform\-Large}$, has an Intel Xeon E5-2683 v4 processor with 16 cores, 40MB L3 cache, and three single-channel 16GB PC-2400 DDR4 DRAMs. We repeated the same methodology as described in the previous experiment. 

Fig.~\ref{fig:synth_profiles_platformlarge} shows the generative profiles overlaid on the raw empirical data for benchmarks $\mathsf{fft}$ and $\mathsf{canneal}$ under a resource context of $\beta=\left(13,11,2.1\right)^\top$ on $\mathsf{PlatformLarge}$. 
\reva{We observe that our generative profiles accurately capture the time-varying execution characteristics for contexts outside of the training set $\mathcal{B}'$ and across different hardware platforms.}

\begin{remark}
    It may seem natural to choose the \emph{mean} value of each conditional joint distribution to construct our generative profiles. Thus, Fig.~\ref{fig:synth_profiles_accuracy} and Fig.~\ref{fig:synth_profiles_platformlarge} 
    also show generative profiles constructed in this way. As can be seen most clearly in the `tail-end' behavior of the mean profiles for $\mathsf{fft}$ in Fig.~\ref{fig:synth_profiles_platformlarge}, however, this choice can lead to inaccurate inference when $\mu_t$ is non-Gaussian. In such cases, first few moments, such as the mean and covariance, are misleading since the statistical typicality differs from the average. As the MSB allows us to obtain the \emph{maximum-likelihood} distributions $\mu_t$ in a nonparametric manner, we can obtain accurate results without any assumptions on the nature of these distributions.
    \label{rmk:maxlikelihood_vs_mean}
\end{remark}

\begin{figure}[t]
    \centering     \vspace{-3ex}
    \includegraphics[width=.75\linewidth]{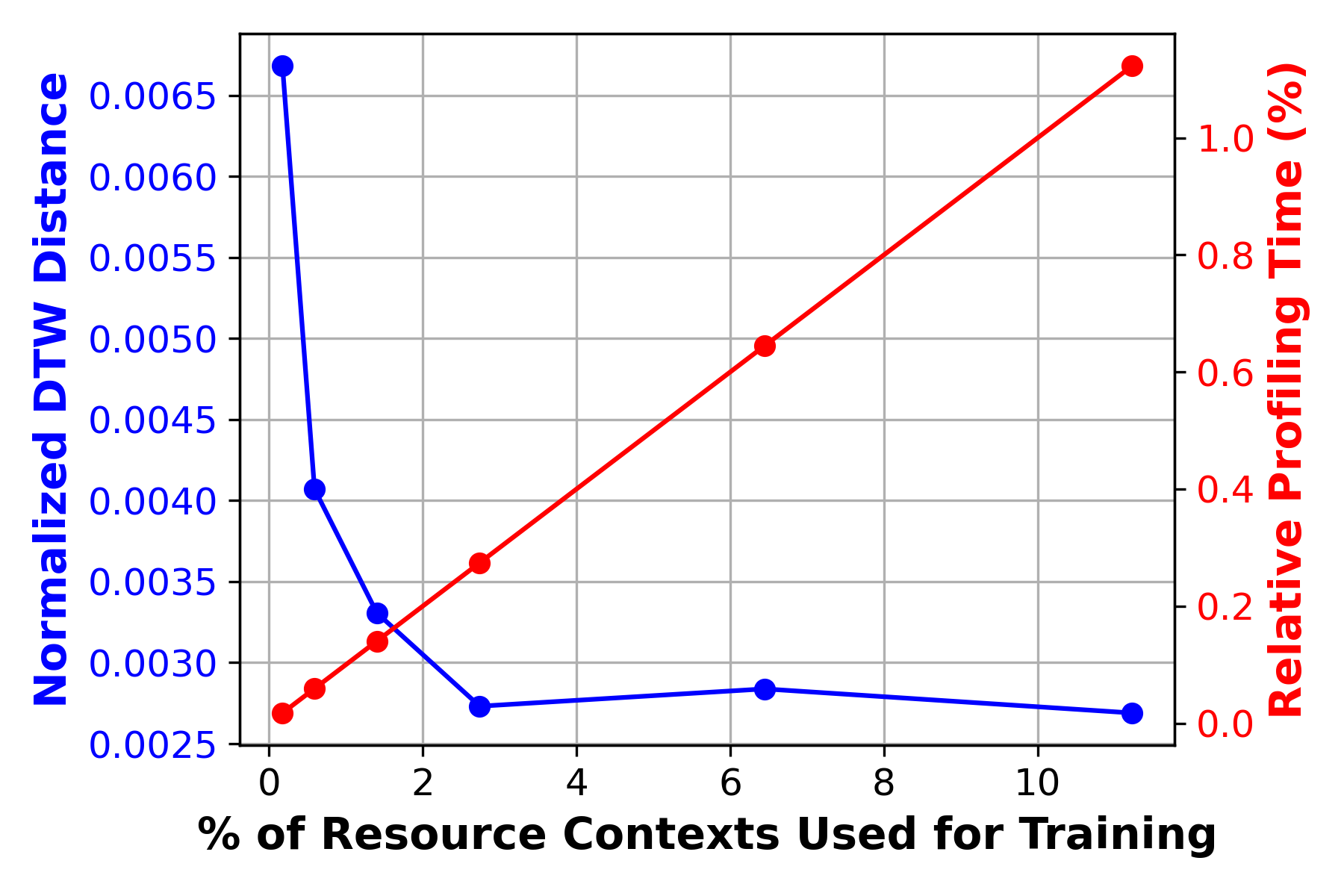}
    \caption{Accuracy and relative profiling time vs. \% of resource contexts used for training (workload $\mathsf{canneal}$).}     \vspace{-1ex}
    \label{fig:accuracy_runtime_tradefoff}
\end{figure}

\noindent{\bf Impact of the number of resource contexts for training.} 
It is difficult to meaningfully assess the absolute accuracy of our generative profiles in a quantitative manner, since the evaluation of any metric on the space of curves in $\mathbb{R}^2$ will depend on the magnitude of the curves themselves. In Fig.~\ref{fig:accuracy_runtime_tradefoff}, we  illustrate the impact of increasing the number of empirical resource contexts used for training (i.e., increasing $n_b$, while keeping $n_d$ unchanged) on the normalized Dynamic Time Warping (DTW) distance \cite{berndt1994using} (shown in \textcolor{blue}{blue}) between the ground-truth empirical profiles and the maximum-likelihood generative profiles for the $\mathsf{canneal}$ benchmark. We use the DTW distance as it is a common metric for evaluating the similarity between time series.
\revg{For each context $\beta$, we computed the DTW distance between the generated profile and the average empirical profile using the $\texttt{fastdtw}$ Python library \cite{salvador2007toward}. To normalize the result, we divided it by the length of the DTW reference path (the average empirical profile)\ multiplied by the maximum norm of the execution state vector $\xi$ along the path.}
We also display the measurement time (in \textcolor{red}{red}) for collecting the training data, normalized relative to that of the full empirical profiles. 

We observe that, as the proportion of resource contexts used for training increases, the DTW distance decreases sharply, while the profiling time increases linearly. The DTW distance  stabilizes at approximately 0.00273 when using only about 3\% of the resource contexts, and adding more contexts beyond this provides little to no benefit. In the case of $\mathsf{canneal}$, this corresponds to only 52 minutes of profiling time compared to approximately 2 weeks to collect the full set of empirical data. In other words, our method can generate accurate profiles for unseen resource contexts using only a tiny fraction of time. 


\begin{table}[t]
  \centering
  \caption{\revg{Per-benchmark accuracy (normalized DTW) of baseline vs. generated profiles, both trained on $\approx 6\%$ of resource contexts, averaged over all $\beta\in\mathcal{B}.$} \reva{Smaller DTW indicates better accuracy. Improvement \% is computed via $100 \times (\textnormal{baseline DTW} - \textnormal{generative DTW})\; / \textnormal{ baseline DTW}$.}}
  \label{tab:all_accuracies}
  \setlength{\tabcolsep}{6pt}
  \renewcommand{\arraystretch}{1}
  \resizebox{.95\columnwidth}{!}{%
    \begin{tabular}{|l|r|r|c|}
      \hline
      \textbf{}           & \textbf{Baseline} & \textbf{Generative} & \textbf{Improvement} \\
      \hline
      \hline
      $\mathsf{blackscholes}$   & 0.0429 & 0.0343 & 20.0\%\\
      \hline
      $\mathsf{bodytrack}$ & 0.0437 & 0.0375 & 14.2\%\\
      \hline
      $\mathsf{canneal}$ & 0.0227 & 0.0027 & 88.1\%\\
      \hline
      $\mathsf{dedup}$ & 0.0321 & 0.0191 & 40.5\%\\
      \hline
      $\mathsf{fft}$ & 0.0508 & 0.0478 & 5.9\%\\
      \hline
      $\mathsf{fluidanimate}$ & 0.0357 & 0.0275 & 23.0\%\\
      \hline
      $\mathsf{radiosity}$ & 0.0404 & 0.0380 & 5.9\%\\
      \hline
      $\mathsf{streamcluster}$ & 0.0549 & 0.0418 & 23.9\%\\
      \hline
    \end{tabular}%
  }
  \vspace{-3ex}
\end{table}

\reva{\noindent\textbf{Accuracy results for all benchmarks and contexts.} Table~\ref{tab:all_accuracies} shows the per-benchmark accuracy (reported as the normalized DTW and averaged over all $\beta\in\mathcal{B}$) of generative profiles compared to a baseline that interpolates between known resource contexts. Specifically, for each unknown context $\beta$, the baseline finds the two resource contexts $\beta', \beta'' \in \mathcal{B}'$ such that $\beta'_i \leq \beta_i \leq \beta''_i \;\forall i \in \{1, 2, \dots, b\}$, and averages their profiles at each snapshot to produce profiles for $\beta$.\footnote{Since the set $\mathcal{B}'$ is chosen such that it contains the minimum and maximum resource contexts,  $\beta'$ and $\beta''$ always exist.} We include this interpolation baseline as it represents a natural approach for estimating execution behavior from sparse resource contexts (for example,~\cite{kpart-hpca-2018} leverages a similar approach to predict execution behavior under sparse LLC allocations). Both the generative profiles and the baseline were trained on approximately 6\% of the resource contexts.

Observe that, for all benchmarks in Table~\ref{tab:all_accuracies}, our generative profiles are closer to the empirical profiles than the baseline in terms of average normalized DTW distance, by up to 88.1\% and by 27.7\% on average. The results demonstrate that our technique more accurately captures the time-varying execution characteristics of each task. 

}


%% file: usecases.tex
\section{Resource Allocation with Generative Profiles}\label{sec:motivation}

\subsection{Case Study: Dynamic Multicore Resource Allocation}\label{subsec:usecase1}
\noindent To demonstrate the practical utility of our technique, we present a representative case study of generative profiles in adaptive multicore resource allocation for real-time systems. Towards this, we propose a frequency-aware extension of DNA~\cite{dna-rtas21}, a phase-aware resource allocation method that dynamically allocates cache and memory bandwidth to tasks at fine-grained time intervals based on their profiles. 
Given a task's profiles, DNA uses clustering to identify a series of phases for each resource context. Each phase is a contiguous sequence of instructions with similar execution metrics (rates of instruction retirement, cache requests and cache misses). At run time, DNA utilizes this context-dependent phase information  to find an allocation that \reva{aims to} maximize the total instruction rate of the system. DNA has been shown to substantially improve schedulability, reduce deadline misses, and decrease latencies compared to static allocation~\cite{dna-rtas21}.

DNA only controls cache and memory bandwidth allocation, while assuming the system operates at a static CPU frequency, typically the maximum frequency. However, as shown in Fig.~\ref{fig:fftwithcontexts}, the instruction rate of a task in a phase is influenced by the CPU frequency, and this relationship varies across different phases and different cache and memory bandwidth allocations. For example, increasing the frequency may improve the instruction rate in a CPU-bound phase but have little to no effect in a memory-bound phase. This presents an opportunity to reduce energy consumption. Therefore, we introduce a phase-based DVFS extension of DNA, which uses task profiles to determine the optimal frequency for each phase under each cache and memory bandwidth allocation. Specifically, it selects the minimum frequency for each phase of a task such that the instruction rate in the phase does not increase by more than a threshold ratio $\epsilon > 0$, compared to the rate achieved under the static maximum frequency $f_{\mathsf{max}}$. 

Define the {\em reference utilization} of a task (or taskset) to be the utilization under the maximum resource context $\beta_{\max} = (N_{\mathsf{ca}},N_{\mathsf{bw}},f_{\mathsf{max}})$. In general, schedulability under global EDF depends on the maximum per-task utilization as well as the total taskset utilization~\cite{goossens2003priority}. Therefore, the smaller these values are for a taskset, the larger we can scale the taskset's threshold $\epsilon$ to enable more power savings.
If the taskset is deemed ``unschedulable'' based on the reference utilization, we set $\epsilon = 0$. Otherwise, we select the largest $\epsilon$ such that the taskset still remains schedulable (based on the reference utilization). For convenience, we refer to this extension of DNA as \dvfs and DNA with static maximum frequency as \static. 

\subsection{Prototype Implementation in Linux}\label{sec:prototype}

\noindent{\bf Overview.} For our experiments, we implemented a prototype of DNA and its DVFS extension on Linux~6.12 with the real-time patch \texttt{PREEMPT\_RT} enabled.\footnote{A prototype of DNA~\cite{dna-rtas21} exists on Xen; however, to avoid virtualization effects, we built our own prototype directly in Linux.}
We implemented DNA as a loadable kernel module, comprising approximately 1{,}900 lines of C code. The module  contains the core DNA algorithm, timer support, mechanisms for applying cache partitions, and calls into MemGuard to enforce memory-bandwidth partitions. In addition, we added about 150 lines to Linux’s in-tree codebase to track DNA metadata within task structures, such as workload type and current phase information.
We also provide user-level runtime scripts to load task phase information, generate task sets, and launch experiments. The run script uses cgroup v2 to confine real-time tasks to a subset of cores and synchronizes release times via the cgroup v2 freezer controller. DNA  is agnostic to the CPU scheduling policy---it simply identifies the current phase of each running task and uses that information to compute the number of cache and memory-bandwidth partitions per core. 

\noindent{\bf DNA invocation mechanism.} To invoke DNA, we used a periodic timer handler that executes on CPU 0 and is configured to fire every 5 ms. Upon each firing, CPU 0 identifies the currently executing task on each experimental core and determines its current phase; if some task has entered a new phase, CPU 0 invokes DNA to compute a new resource allocation and then applies it for all cores. In addition, we modified Linux's context switch mechanism to invoke DNA whenever a new task is assigned to a CPU, to ensure that a newly running task always receives a correct allocation.

\noindent{\bf Resource and frequency assignment.} We leveraged Intel's CAT hardware technology~\cite{intel-2015-cat} and MemGuard~\cite{Yun16-memguard-journal} for cache and memory bandwidth allocation. Using CAT, we can efficiently apply a cache allocation to CPUs by writing to MSR registers.
In contrast, MemGuard relies on hardware performance counters to monitor the number of L3 cache misses on each CPU over fixed time periods; if a CPU's number exceeds its assigned budget, MemGuard throttles it by scheduling a ``MemGuard'' thread, which spins until the next replenishment period. 
To work correctly, the ``MemGuard'' threads must have the highest priority, which does not hold under  \texttt{SCHED\_DEADLINE}. To address this, we modified MemGuard's throttle handler to avoid scheduling the ``MemGuard'' thread. Instead, it sets a special throttled bit on the affected CPU, triggers a reschedule on that CPU, and returns. The Linux scheduler was updated to detect this bit: when set, the \texttt{SCHED\_DEADLINE} scheduler class is bypassed, so that MemGuard's throttle thread can run until the bit is cleared, after which normal scheduling resumes. 
Additionally, we adapted MemGuard to work under  \texttt{PREEMPT\_RT}.

To enable frequency scaling, we extended the DNA module to set the CPU frequency each time DNA is invoked. This enables us to set the CPU frequency at the start of each phase, as is given by \dvfs for the current cache and memory-bandwidth allocation (c.f. Sec.~\ref{subsec:usecase1}). 

\noindent{\bf Run-time overhead.} Table~\ref{tab:dna-overhead} shows the runtime overhead. On average, it takes only about \(1\,\mu\text{s}\) to compute an allocation across all cores using our prototype, which is minimal.

\begin{table}[t]
  \centering
  \caption{Runtime Overhead of DNA in \(\mu\)s}
  \label{tab:dna-overhead}
  \setlength{\tabcolsep}{6pt}
  \renewcommand{\arraystretch}{1}
  \resizebox{.85\columnwidth}{!}{%
    \begin{tabular}{|c|c|c|c|c|c|}
      \hline
      \textbf{Count} & \textbf{Mean} & \textbf{90th} & \textbf{95th} & \textbf{99th} & \textbf{Max} \\
      \hline
      28347 & 1.068 & 1.536 & 2.008 & 6.204 & 10.727 \\
      \hline
    \end{tabular}%
  }
  \vspace{-3ex}
\end{table}

We next present an experimental evaluation of generative profiles using this prototype and the discussed case study.

%% file: sched-eval.tex
\section{Experimental evaluation}\label{sec:sched-eval}

\subsection{Experimental Setup}

\noindent{\bf Tasksets.}\label{sec:dna-workload}
We generated 20 tasksets, with taskset reference utilizations ranging between 0.2 and 4.0, at steps of 0.2. 
\reva{To generate a taskset, we iteratively added tasks, each with a utilization uniformly sampled from $[0.1, 0.4]$, into the taskset until its total utilization is at least the target utilization. The final task’s utilization was adjusted such that the taskset utilization equals its target utilization.}
For each task, we then randomly picked one of the benchmarks in the PARSEC suite as its workload. The task's profiles are obtained from the measurements discussed in Sec.~\ref{sec:workload-measurement}. We set the task's period (relative deadline) to be its {\em reference} WCET---defined as the observed WCET under the maximum resource context $\beta_{\max} = (N_{\mathsf{ca}},N_{\mathsf{bw}},f_{\mathsf{max}})$---divided by its utilization. 

\noindent{\bf Experiments.} 
We executed each taskset under \texttt{SCHED\_DEADLINE} on 4 experimental cores of our default multicore platform, with DNA prototype running. Each task's budget was set equal to its period.  Jobs were released periodically based on their periods over a one-minute interval and executed to completion. We recorded each job's release and completion time to compute its response time. We repeated this experiment for all generated tasksets. In total, we evaluated 400 unique tasksets. 

\noindent{\bf Methods under evaluation.} We considered three different inputs to our resource allocation algorithm: (1) complete {\em empirical} profiles; (2) our {\em generative} profiles learned from a small subset $\mathcal{B}'$ (containing  0.27\%) of the empirical profiles; and (3) {\em baseline} profiles, which consist of the same subset $\mathcal{B}'$ of empirical profiles used to produce our generative profiles, but with the profiles for each unknown context created using the simple interpolation method \reva{described in the accuracy evaluation of Sec.~\ref{sec:accuracy}.}
 
 \noindent{\bf Metrics.} For each method, we evaluated the {\em observed} schedulability and energy savings under \dvfs. For comparison, we also evaluated \static  with the highest CPU frequency setting ($f_{\mathsf{max}} = 2.3$~GHz),  using the full empirical profiles. To capture the effect of dynamic frequency scaling, we measured the {\em energy savings} based on dynamic power consumption, obtained by subtracting the static power consumption (measured when all CPUs were idle) from the total power consumption reported by Intel's \texttt{RAPL} energy registers~\cite{intel_dev_manual}. 
 The energy consumption for an experiment run was calculated as the product of the average CPU dynamic power consumption and the duration of the experiment. 
 

\begin{figure}[t]
\centering
\includegraphics[width=0.8\columnwidth]{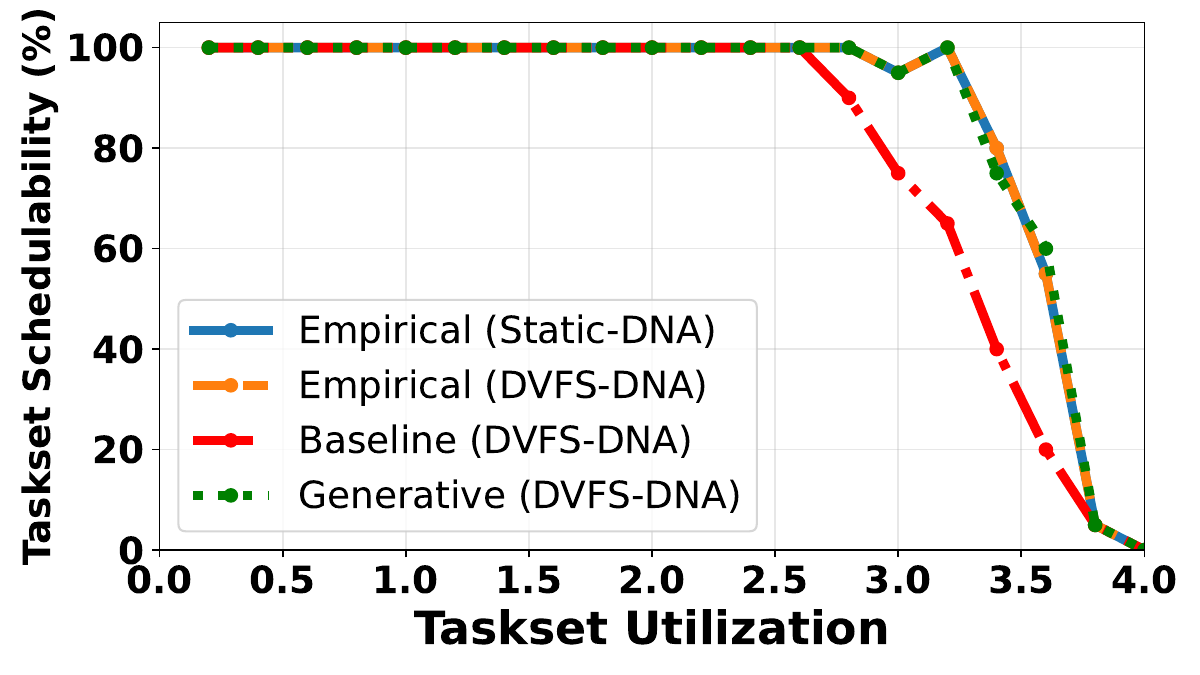}\vspace{-2ex}
\caption{Taskset Schedulability}
\label{fig:taskset_schedulability}
\vspace{-1ex}
\end{figure}

\begin{table}[t]
  \centering
  \caption{Avg. End-to-end Measurement \& Generation Time}
  \label{tab:profile_time}
  \setlength{\tabcolsep}{6pt}
  \renewcommand{\arraystretch}{1}
  \resizebox{.85\columnwidth}{!}{%
    \begin{tabular}{|l|r|r|}
      \hline
      \textbf{}           & \textbf{Total Time (hrs)} & \textbf{\# Empirical Profiles} \\
      \hline
      \textbf{Empirical}  & 231 & 456,000 \\
      \hline
      \textbf{Baseline}   & 0.64     & 1,250   \\
      \hline
      \textbf{Generative} & 1.14     & 1,250   \\
      \hline
    \end{tabular}%
  }
  \vspace{-3ex}
\end{table}


\begin{figure*}[t]
    \centering
    \includegraphics[width=\textwidth]{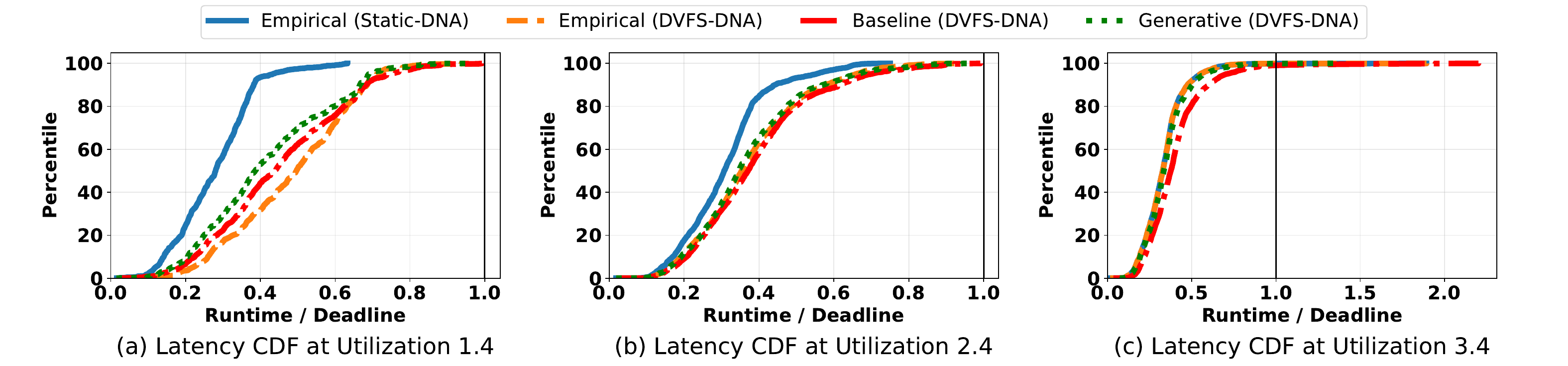}
    \caption{DNA Performance Comparison using Generative vs. Baseline vs. Empirical Profiles.}\vspace{-3ex}
    \label{fig:CDF_results}
\end{figure*}

\subsection{Results}
\noindent{\bf Schedulability results.}\label{sec:dna-schedulability}
Fig.~\ref{fig:taskset_schedulability} shows the percentages of schedulable tasksets across different taskset utilizations for \dvfs under the three input methods (generative, baseline, empirical), as well as \static running with full empirical profiles.  
We observe high schedulability for both empirical and generative profiles. For instance, the first unschedulable taskset only occurs at $3.0$ utilization (i.e., 75\% CPU load). The results highlight the benefits of context-dependent fine-grained profiles in dynamic resource allocation. 

Notably, the performance of generative profiles closely matches the empirical profiles' performance, demonstrating that our generative profiling is effective in practical applications. Table~\ref{tab:profile_time} further demonstrates the average end-to-end time needed to measure empirical and produce generative profiles per benchmark. Notice that  the generative profiling method requires just over an hour to collect training data and generate profiles per benchmark, compared to 231 hours (9.6 days) to obtain the complete empirical profiles. In other words, it is 230$\times$ more efficient while delivering comparable schedulability compared to using complete empirical profiles.   

Finally, we observe that generative profiles substantially outperform baseline profiles in terms of schedulability. Overall, the generative profiling method schedules 1.63$\times$ more tasksets than the baseline between 3.0 and 4.0 utilization. This shows that profiling a limited number of resource contexts and filling in the gaps with a simple interpolation method is insufficient to achieve equivalent real-time performance (compared to full empirical profiles). \\[-2ex]

\noindent{\bf Response time results.}\label{sec:dna-latency}
Fig.~\ref{fig:CDF_results} shows the CDFs of job response time over deadline at three different representative taskset utilizations (low, medium and high). The vertical line at 1.0 represents the point at which a job's response time equals its deadline. Again, we observe that our generative profiles achieve nearly identical job response times as empirical profiles do under \dvfs, and the results are consistent across different taskset utilizations. 

Notice that the job response times when using empirical profiles under \static are overall smaller than those obtained under \dvfs. This is expected since \static uses the maximum frequency setting while \dvfs aims to save energy. As the utilization increases, however, the gap between them becomes closer, indicating \dvfs's ability to scale the frequency to workload demands. On the contrary, the baseline method experiences noticeably longer tail response times, especially at higher utilizations. This is consistent with the observed schedulability results and further highlights the limitation of baseline profiles. \\[-2ex]

\noindent {\bf Energy saving results.}
We also evaluated the dynamic energy consumption of \dvfs under each profiling method compared to that of \static with full empirical profiles. The results show that \dvfs substantially reduces the energy consumption for all methods  compared to \static. On average, \dvfs  reduces the total energy consumption by $18.1$\% under empirical profiles and $14.8$\% under generative and baseline profiles. More importantly, for both generative and  empirical profiles, these savings are achieved while maintaining the same schedulability performance as \static (described earlier). The results demonstrate the needs and benefits of considering all three resource types (frequency, cache, memory bandwidth) in real-time resource allocation, which can now be achieved efficiently with generative profiles.

%% file: related.tex
\section{Related Work}\label{sec:related}

\noindent\textbf{Context-dependent profiles.}
Context-dependent profiles are increasingly being used in multicore real-time systems. For example, existing techniques have utilized resource-aware task WCETs to perform task mapping and shared resource allocation~\cite{xu-rtas19,xu2019holistic}.
As described in Section \ref{subsec:usecase1}, DNA (and its deadline-aware variant DADNA)~\cite{dna-rtas21} further exploit fine-grained  profiles to dynamically adapt resources allocated to a task to maximize execution progress, while~\cite{rasco-emsoft25} co-designs fine-grained resource allocation with scheduling. Likewise, recent work in real-time DVFS~\cite{rt-dvfs, hrt-dvfs} has started to explore phase-based or workload-aware DVFS techniques~\cite{acun2019dvfs,powerock,dpdfvs,SWEEP}. 
\reva{Other work leverages context-dependent profiles to improve general system performance and/or fairness via resource allocation~\cite{park-2019-copart, clite-hpca-2020, Pons-TPDS-2020, satori-isca-2021}. However, the above techniques often require extensive profiling of each task under each resource context to make allocation decisions. For multicore systems with many tasks and multiple resource types, such extensive profiling requirements can prohibit the adoption of these techniques.

One recent solution for generating context-dependent profiles uses limited empirical profiling to estimate WCETs under unseen memory bandwidth allocations~\cite{Sohal-RTS-2022}. However, its analysis requires detailed knowledge of the throttling mechanism (e.g., MemGuard), fixes all other resources, and focuses on memory transactions. Since our goal is to produce a maximum likelihood profile (not a worst-case timing analysis), our method is more general and directly predicts, without such restrictions, the evolution of a program’s \emph{multi-dimensional} micro-architectural state under \emph{multi-dimensional} unseen resource contexts.

}

\noindent\textbf{Learning-based prediction in scheduling.}
There is a substantial body of work that applies statistical methods to estimate probabilistic response times or deadline misses~\cite{Lu10,Lu12,Liu13,LuYue10}, although much of this work is restricted to the prediction of course-grained information such as job size and end-to-end response time.
Recently, Marković et al.~\cite{Markovic24} also used nonparametic learning to estimate statistical upper bounds on the mean, standard deviation and covariance of task execution time. In contrast to their work, we consider the much finer-grained microarchitectural execution characteristics and the time-varying stochastic relationship between multidimensional resource usage, conditioned on an allocated resource context. Prior work \cite{bondar-2024-psmsbp,bondar-2024-stochastic} considered the former, but did not address the dynamic correlation between a software's stochastic time-varying resource usage and the resources allocated to the software, as we do in this work.\\
\noindent\textbf{MSB.} The formulation and usage of MSB for statistical inference and learning, is a relatively recent endeavor \cite{elvander2020multi,noble2023tree,chen2023deep,bondar-2024-psmsbp,bondar-2024-stochastic}. The MSB can be understood as the stochastic version of the multi-marginal optimal transport \cite{pass2015multi}, which corresponds to vanishing entropic regularization ($\varepsilon\downarrow 0$). The convergence of the multi-marginal Sinkhorn algorithms in the continuous state space was established in \cite{marino2020optimal,carlier2022linear}. 

%% file: supplementaryRTAS2026.tex
\onecolumn

\begin{center} 
\LARGE{\bf{Supplementary Materials}} \end{center}

\vspace*{0.2in}


\renewcommand{\thesection}{S\arabic{section}} 
\setcounter{section}{0} 


\section{Maximum Likelihood Guarantee for MSBP}
The minimizer of \eqref{DiscreteMSBP}, denoted as $\bm{M}_{\rm{opt}}$, is guaranteed to be the statistically most parsimonious coupling that is consistent with the distributional snapshot (i.e., scattered point cloud) data. To see this explicitly, recall that for two given probability mass tensors $\bm{P},\bm{Q}$, the relative entropy a.k.a. the Kullback-Leibler divergence ${\mathrm{D_{KL}}}(\bm{P}\parallel\bm{Q}):=\langle\bm{P},\log\bm{P}-\log\bm{Q}\rangle$, where $\log$ acts elementwise. Let 
\begin{align}
\bm{K}:=\exp(-\bm{C}/\varepsilon), \quad Z := \!\!\sum_{i_1,\hdots,i_{n_{s}}}\!\!\left[\bm{K}_{i_1,\dots,i_{n_s}}\right] > 0,
\label{defK}    
\end{align}
and define the \emph{Gibbs distribution}
\begin{align}
\bm{M}_{{\mathrm{Gibbs}}} := \bm{K}/Z.
\label{GibbsDistribution}    
\end{align}
Theorem \ref{ThmMoptMinimizesRelativeEntropy} 
quantifies the qualifier ``most parsimonious" in terms of relative entropy optimality certificate.
\begin{theorem}[Relative entropy optimality]\label{ThmMoptMinimizesRelativeEntropy}
The MSB $\bm{M}_{\rm{opt}}$ is the unique minimizer of the relative entropy $\bm{M}\mapsto{\mathrm{D_{KL}}}(\bm{M}\parallel\bm{M}_{{\mathrm{Gibbs}}})$ subject to the observational constraints \eqref{DiscereteMSBPconstr}.    
\end{theorem}
\begin{proof}
Using the definition of ${\mathrm{D_{KL}}}$ and \eqref{GibbsDistribution}, we find
\begin{align}\label{DKL1}
&{\mathrm{D_{KL}}}(\bm{M}\parallel\bm{M}_{{\mathrm{Gibbs}}}) = \langle\bm{M},\log\bm{M}\rangle -\langle\bm{M},\log\bm{K}\rangle + \log Z\!\!\!\!\!\!\underbrace{\sum_{i_1,\hdots,i_{n_{s}}}\left[\bm{M}_{i_1,\dots,i_{n_s}}\right]}_{=1,\,\text{since $\bm{M}$ is prob. mass tensor}}\!\!\!\!.
\end{align}
Using \eqref{defK}, we then re-write \eqref{DKL1} as
\begin{align}\label{DKL2}
&{\mathrm{D_{KL}}}(\bm{M}\parallel\bm{M}_{{\mathrm{Gibbs}}}) = \frac{1}{\varepsilon}\langle\bm{C}+\varepsilon\log\bm{M},\bm{M}\rangle + \log Z.
\end{align}
Notice that \eqref{DKL2} matches with the objective \eqref{DiscreteMSBPobj} up to translation and scaling by constants. However, translation and scaling of the objective cannot change the $\arg\min$. Since $\bm{M}_{\rm{opt}}$, by definition, satisfies the primal feasibility \eqref{DiscereteMSBPconstr}, $\bm{M}_{\rm{opt}}$ is a minimizer of $\bm{M}\mapsto{\mathrm{D_{KL}}}(\bm{M}\parallel\bm{M}_{{\mathrm{Gibbs}}})$ subject to the constraints \eqref{DiscereteMSBPconstr}. The uniqueness follows from the strict convexity of ${\mathrm{D_{KL}}}$ with respect to its first argument \cite[p. 235]{niculescu2025convex}. 
\end{proof}

\begin{remark}[Geometric interpretation]
Theorem \ref{ThmMoptMinimizesRelativeEntropy} has a geometric interpretation: the MSB $\bm{M}_{\rm{opt}}$ is the Kullback-Leibler projection of $\bm{M}_{\rm{Gibbs}}$ at the intersection of the hyperplanes defined by \eqref{DiscereteMSBPconstr}. This is a special case of information projection \cite{csiszar2003information}, and can be solved by alternating Kullback-Leibler projection to each hyperplane: an idea that goes back to Von Neumann \cite{von1949rings}; see also \cite{bregman1967relaxation,censor1994multiprojection}.
\end{remark}

\begin{remark}[From exponential to linear complexity]
We note that the MSBP \eqref{DiscreteMSBP} has exponential complexity, and is not amenable to a direct numerical solution. In Sec. \ref{subsec:DulaityForMSBP}, we explain how duality can be leveraged to circumvent this computational difficulty. In particular, in Sec. \ref{subsec:FromConditionalMSBPtoGenerativeProfiling}, we point out that the proposed algorithm has linear complexity in $n_{s}$, whereas \eqref{DiscreteMSBP} had complexity exponential in $n_{s}$.
\end{remark}

\begin{remark}[Understanding the maximum-likelihood guarantee] 
Let us explain the maximum-likelihood guarantee with respect to the schematic in Fig. \ref{fig:conceptschematic}. The guarantee in Theorem \ref{ThmMoptMinimizesRelativeEntropy} refers to that of the distributions in the bottom left of Fig. \ref{fig:conceptschematic}, i.e., these are the most probable distributions among all possible distributions consistent with the observed profiles. Once these distributions are learnt, the highest probability sample paths are drawn from these learnt distributions, resulting in the maximum-likelihood sample paths of $\xi(t)\mid\beta$ as shown in the bottom right of Fig. \ref{fig:conceptschematic}. So there are in fact two maximum-likelihood guarantees: the first is for the distributions, and the second is for the highest probability sample path given the distributions. The former notion is highly non-trivial while the latter is a standard statistical construct. Conditional MSBP is not an ad hoc algorithm, it is the unique solution of the former.

To further appreciate this idea, consider what happens if one replaces the conditional MSBP with some heuristic. For instance, one may interpolate (using some fixed interpolation scheme) the observed snapshots $\mu_1$ at $t_1$ and $\mu_2$ at $t_2$ to form $\mu_{\sigma}$ at any $t_{\sigma}\in(t_1,t_2)$, and then, as in our case, draw the highest probability sample from each such $\mu_{\sigma}$, constructing the sample path $\xi(t)\mid\beta$. Such a sample path does not have the maximum likelihood guarantee even tough the second step is identical to ours! This is because the maximum likelihood guarantee will then be lost in the first step at the distribution level.
\end{remark}


\section{Computation of Unimarginal and Bimarginal Projections}

Efficient computation of the unimarginal projection \eqref{DefUnimargProj} and the bimarginal projection \eqref{DefBimargProj} are made possible by exploiting the path structure of $\bm{C}$ given by \eqref{CostTensorPathStructure}, which dictates the structure of $\bm{K}$, and hence the structure of $\bm{K}\odot\bm{U}$. Specifically, we compute these projections as follows.

\begin{proposition}(\cite[Proposition 1]{bondar-2024-psmsbp})\label{Prop:MarginalExplicit}
Given a path-structured cost tensor $\bm{C}$ satisfying \eqref{CostTensorPathStructure}, and $\varepsilon>0$, consider the tensors $\bm{K}$ as in \eqref{defK} and $\bm{U}$ as in \eqref{Moptstructure}. Define matrices $K^{\sigma}:=\exp\left(-C^{\sigma}/\varepsilon\right)$ $\forall \sigma\in\llbracket n_{s}\rrbracket$. Then, the unimarginal projection
\begin{align}
\operatorname{proj}_\sigma(\boldsymbol{K} \odot \boldsymbol{U})=\left(u_1^{\top} K^{1 } \prod_{j=2}^{\sigma-1} \operatorname{diag}\left(u_j\right) K^{j}\right)^{\!\!\top}\!\!\odot u_{\sigma} \odot
\left(\left(\prod_{j=\sigma+1}^{{n_{s}}-1} K^{j-1} \operatorname{diag}\left(u_j\right)\right) K^{{n_{s}}-1} u_{n_{s}}\right) \;\;\forall \sigma \in \llbracket n_s \rrbracket,
\label{UnimargExplicit}
\end{align}
and the bimarginal projection
\begin{align}
& \operatorname{proj}_{\sigma_1, \sigma_2}(\boldsymbol{K} \odot \boldsymbol{U})= \operatorname{diag}\left(u_1^{\top} K^{1} \prod_{j=2}^{\sigma_1-1} \operatorname{diag}\left(u_j\right) K^{j}\right)\operatorname{diag}\left(u_{\sigma_1}\right) \!\prod_{j=\sigma_1+1}^{\sigma_2}\!\left(K^{j-1} \operatorname{diag}\left(u_j\right)\!\right) \nonumber\\
& \qquad\qquad\qquad\qquad\quad\operatorname{diag}\left(\!\left(\prod_{j=\sigma_2+1}^{{n_{s}}-1}\!\!K^{j-1} \operatorname{diag}\left(u_j\right)\right) K^{{n_{s}}-1} u_{n_{s}}\!\right)\qquad\forall\left(\sigma_1, \sigma_2\right) \in\left\{\llbracket n_s \rrbracket^{\otimes 2} \mid \sigma_1<\sigma_2\right\}.
\label{BimargExplicit}
\end{align}
\end{proposition}
Thanks to the above, computing \eqref{UnimargExplicit} and \eqref{BimargExplicit} are vectorized in the sense they only involve matrix-vector multiplications. \revg{We note while $\bm{C}$ in our case follows a natural path structure, as distributions $\mu_\sigma$ are ordered sequentially by time, this is not necessary for efficient computation of the projections. For instance, \cite{haasler2021multimarginal} shows that if the $\mu_\sigma$ are correlated via a tree graph, a similar result as above holds. Indeed, \cite{bondar2025optimal} showed that trees are the optimal (in the cost-minimizing sense) graph structure for the general class of graph-structured SBPs.} \newline